\documentclass[amsmath,amssymb]{revtex4-2} 
\usepackage{hyperref}
\usepackage{geometry}
\usepackage{dsfont}
\usepackage{graphicx}
\usepackage{tabularx}
\usepackage[x11names]{xcolor}
\usepackage{amsthm}
\usepackage{algorithm}
\usepackage{algpseudocode}
\usepackage{enumitem}
\usepackage[caption=false]{subfig}
\usepackage{tikz}
\def \globalscale {0.280000}
\usepackage{overpic}
\usepackage{placeins}
\usetikzlibrary{quantikz}
\colorlet{color1}{Purple3}
\colorlet{color2}{DeepSkyBlue4}
\colorlet{color3}{Yellow4}
\newcommand{\vqwa}[1]{
	\edef\start{\the\pgfmatrixcurrentrow-\the\pgfmatrixcurrentcolumn}
	\edef\end{\the\numexpr#1+\pgfmatrixcurrentrow\relax-\the\pgfmatrixcurrentcolumn}
\expandafter\expandafter\expandafter\vqwexplicita\expandafter\expandafter\expandafter{\expandafter\start\expandafter}\expandafter{\end}
}
\newcommand{\vqwexplicita}[2]{
\arrow[from=#1,to=#2,arrowsa] {}
}
\newcommand{\vqwb}[1]{
\edef\start{\the\pgfmatrixcurrentrow-\the\pgfmatrixcurrentcolumn}
\edef\end{\the\numexpr#1+\pgfmatrixcurrentrow\relax-\the\pgfmatrixcurrentcolumn}
\expandafter\expandafter\expandafter\vqwexplicitb\expandafter\expandafter\expandafter{\expandafter\start\expandafter}\expandafter{\end}
}
\newcommand{\vqwexplicitb}[2]{
\arrow[from=#1,to=#2,arrowsb] {}
}
\newcommand{\vqwc}[1]{
\edef\start{\the\pgfmatrixcurrentrow-\the\pgfmatrixcurrentcolumn}
\edef\end{\the\numexpr#1+\pgfmatrixcurrentrow\relax-\the\pgfmatrixcurrentcolumn}
\expandafter\expandafter\expandafter\vqwexplicitca\expandafter\expandafter\expandafter{\expandafter\start\expandafter}\expandafter{\end}
}
\newcommand{\vqwexplicitca}[2]{
\arrow[from=#1,to=#2,arrowsc] {}
}

\newcommand{\qwa}{\ifthenelse{\the\pgfmatrixcurrentcolumn>1}{\arrow[arrowsa]{l}}{}}
\newcommand{\qwb}{\ifthenelse{\the\pgfmatrixcurrentcolumn>1}{\arrow[arrowsb]{l}}{}}
\newcommand{\qwc}{\ifthenelse{\the\pgfmatrixcurrentcolumn>1}{\arrow[arrowsc]{l}}{}}

\tikzcdset{arrowsa/.style={dash,thick,color1},arrowsb/.style={dash,thick,color2},arrowsc/.style={dash,thick,color3}}
\tikzset{internal/.style={thickness}}

\def\ctrla#1{\control[fill=color1]{}	\vqwa{#1}}
\def\ctrlb#1{\control[fill=color2]{}	\vqwb{#1}}
\def\ctrlc#1{\control[fill=color3]{}	\vqwc{#1}}

\newcommand{\phasea}{\phase[fill=color1]}
\newcommand{\phaseb}{\phase[fill=color2]}
\newcommand{\phasec}{\phase[fill=color3]}

\newcommand{\targa}{\targ[draw={color1}]{}}
\newcommand{\targb}{\targ[draw={color2}]{}}
\newcommand{\targc}{\targ[draw={color3}]{}}

\DeclareMathOperator*{\argmax}{arg\,max}
\DeclareMathOperator{\tr}{tr}

\def\1{\mathds{1}}

\newcommand{\idealmeasurement}{s}

\newcommand{\expectationinsyndrome}{\widetilde{\langle s\rangle}}
\newcommand{\expectationforestimation}{\langle \idealmeasurement\rangle}
\newcommand{\firstexpectationvalue}{\widetilde{\langle s \rangle}_{\mathrm{cal.}}^{(1)}}
\newcommand{\secondexpectationvalue}{\widetilde{\langle s \rangle}_{\mathrm{cal.}}^{(2)}}

\newcommand{\expectationcalibrationideal}{\langle \idealmeasurement \rangle_{\mathrm{cal.}}}
\newcommand{\thestabilizergroup}{\mathcal{S}}
\newcommand{\calibrationinitialstate}{\ket{\psi}}



\newcommand{\bC}{\mathbb C}
\newcommand{\bF}{\mathbb F}

\newcommand{\cE}{\mathcal E}
\newcommand{\cB}{\mathcal B}
\newcommand{\cF}{\mathcal F}

\newcommand{\expvalideal}[2]{\langle #1\rangle_{#2}}
\newcommand{\expvalnoisy}[2]{\widetilde{\langle #1\rangle}_{#2}}

\newcommand{\linstab}{\bC\{\widehat S\}}
\newcommand{\ftwo}{\mathbb F_2}
\newcommand{\blank}{{-}}

\DeclareMathOperator{\Id}{Id}
\DeclareMathOperator{\anc}{anc}
\DeclareMathOperator{\ev}{ev}

\DeclareMathOperator{\avg}{avg}
\DeclareMathOperator{\Syn}{Syn}

\theoremstyle{definition}
\newtheorem{definition}{Definition}
\newtheorem{notation}[definition]{Notation}
\newtheorem{example}[definition]{Example}

\theoremstyle{plain}
\newtheorem{thm}[definition]{Theorem}
\newtheorem{prop}[definition]{Proposition}
\newtheorem{lemma}[definition]{Lemma}
\newtheorem{cor}[definition]{Corollary}
\theoremstyle{remark}
\newtheorem{rem}[definition]{Remark}

\begin{document}
\title{Calibration of Syndrome Measurements in a Single Experiment}
\author{Christian Wimmer}
\author{Jochen Szangolies}
\author{Michael Epping}
\email{Michael.Epping@dlr.de}
\affiliation{Institute for Software Technology, German Aerospace Center (DLR), Cologne, Germany}
\date{\today}
\begin{abstract}
		Quantum error correction can reduce the effects of noise in quantum systems, e.g. in metrology or most notably in quantum computing. Typically, this requires making measurements that provide information about the errors that have occurred in the system.
		However, these syndrome measurements themselves introduce noise into the system, for example by using noisy gates.
		A complete characterization of the measurements is very costly.
		Here we describe a calibration method to obtain the syndrome statistics taking into account the additional noise sources.
		All calibration data are extracted from a single experiment in which the syndrome measurement is performed twice in a row.
		Thus, our method allows an accurate evaluation of syndrome measurements with significantly less effort than existing methods. We give examples of the application of this method to noise estimation and error correction. Finally, we discuss the results of experiments performed on an IBM quantum computer.
\end{abstract}
\maketitle

\section{Introduction}
\label{sec:introduction}
Current quantum computing hardware is approaching a level of maturity where it begins to benefit from the techniques developed in quantum error correction~\cite{Schindler2011,Negnevitsky2018}.
These also have broad applications beyond quantum computing, e.g. in quantum communication~\cite{Ekert1996} and metrology~\cite{Duer2014, Kessler2014}.
A particularly well-studied method is stabilizer codes with syndrome measurements, which provide the information necessary to correct errors~\cite{gottesman1997stabilizer, Shor95, Calderbank96,Steane96}.
However, the measurements usually require demanding multi-qubit gates and themselves introduce noise into the system, partly invalidating the obtained information. 
Therefore, it is of great interest to explore techniques that mitigate such influences, e.g. via machine learning methods~\cite{convy2022machine} or by using an estimate of the error rate, obtained via the syndrome measurements themselves, to inform the correction operation~\cite{huo2017learning}.
Along similar lines, we present a technique to reduce deleterious measurement effects without modifying the circuit using only one additional calibration experiment, 
which involves repeating the measurement twice.
By single experiment we mean a single input state, a single circuit and a single measurement setting which is repeated sufficiently often to estimate the outcome statistics with the desired accuracy.
One inspiration for our approach is the method of Wagner \textit{et. al.} \cite{wagner2022pauli}, where the authors also only use information that is already measured in the context of quantum error correction.

We make the following assumptions. 
The noise introduced by the measurement circuit can be described by Pauli noise and is independent of where the measurement is used inside a larger circuit.
Note that this model covers various types of noise sources. 
For example, it includes initialization errors, coherent and incoherent single and multi-qubit gate errors, cross-talk, readout errors, correlated multi-qubit errors, noise on idle qubits, and error propagation within the circuit.
Any ancillary qubits used in the stabilizer measurement may be subject to the same kind of noise.
Also note that Pauli noise can be enforced by randomized compiling~\cite{Wallman16}, so it is not a strong limitation.

Our main result is that under the described assumptions it is possible to exactly calculate the expectation value of any stabilizer element after the noisy measurement was performed. 
We consider two different cases: 
Averaging over the measurement outcomes or conditioning on individual results.

\begin{figure*}
	\begin{quantikz}
		& \gate{\varepsilon} \qwbundle[alternate]{} &   \meter{Syndrome}\arrow[d]\qwbundle[alternate]{} &\qwbundle[alternate]{}\slice[style={line width=0.5mm}]{} &  \qwbundle[alternate]{}&\gate{\text{corr.}} \qwbundle[alternate]{} & \qwbundle[alternate]{}\\
		& &  \widetilde{\langle s \rangle}\arrow[rr,"\alpha_s",shorten=1pt] & &\vphantom{ \widetilde{\langle s \rangle}}\langle s \rangle & &
	\end{quantikz}		
		\caption{The general scenario in which we are interested: A noise channel $\varepsilon$ is followed by a noisy syndrome measurement which informs the correction operation (corr.). We calculate the expectation values $\langle s \rangle$ for stabilizer elements $s$, right before the correction operation (dashed line), via the correction factor $\alpha_s$, see Eq.~(\ref{eq:correctedexpectationvalue}). This allows estimating the accumulated noise at that position, which is the relevant noise for the error correction.}\label{fig:scenario}
\end{figure*}

In principle process tomography allows to fully characterize the noisy measurement~\cite{dariano2003quantum, Dariano04, Artiles05}. 
However, it involves preparing several different input states and performing multiple additional measurements. 
The tomography is further complicated by consistency issues, see e.g.~\cite{Merkel2013,Cattaneo2023}.
Moreover, these methods produce a more detailed description than necessary in the context of quantum error correction. 

As mentioned above, syndrome measurements do not only involve basic measurements but also gates.
When it comes to characterizing the noise of gates, randomized benchmarking has proven to be very successful~\cite{Emerson_2005}.
This refers to various protocols that are intended to determine the quality of gates as independently as possible of preparation and measurement errors.
It quantifies the averaged error over a set of quantum gates by performing random circuits that are assembled in such a way that overall the identity would be performed if the gates were perfect.
However, when applied to our scenario, only partial information about the syndrome measurement is obtained.

Several known techniques can partially remove noise from a measurement result of short circuits. These approaches are called quantum error mitigation (QEM), and they aim to reproduce the results of noiseless circuits without quantum error correction codes via post-processing of the data \cite{cai2023quantum,endo2019a}.
Zero noise extrapolation~\cite{Temme2017} is an error mitigation method in which a circuit is executed several times with gates of varying degrees of noise. The noise is artificially influenced, e.g. by varying the execution time of the gates or by extending the circuit. The behavior of a noise-free circuit is then extrapolated from the data points obtained.
Readout error mitigation \cite{Tannu2019,beisel2022} is a method that aims to remove measurement errors from the measurement data. A matrix is set up containing the probabilities that one observes a measurement result x instead of the correct one y. The basic idea of the approach is then to apply the inverse of this matrix to improve the statistics.

While we use readout error mitigation for reference data in our experiment, we emphasize that the goal of our proposed method is not to remove the noise from the measured data but on the contrary to take the noise introduced by the measurement into account. 
Furthermore, as will become clear below, the computational complexity of our proposed correction only scales (linearly) with the size $2^m$ of the stabilizer of the error correction code.
In particular it does not scale with the number of physical or logical qubits in the quantum computer.
Thus, our method is also interesting beyond the noisy intermediate-scale quantum (NISQ) era.

In the context of quantum error correction currently the most common approach to deal with noisy syndrome measurements is to include additional redundancies.
This can be done by introducing redundant stabilizer measurements to encode the syndrome using a classical error correction code.
One can then correct errors on the syndrome. 
The concatenation of one code for the data and one for the syndrome is called a data-syndrome code~\cite{Ashikhmin2019}.
It is also common to repeat the syndrome measurement an appropriate number of times and analyze which detected errors may be artifacts of a faulty readout~\cite{Dennis2002}.
Both of these strategies have a higher resource overhead than our proposed method and do not consider calibration data to estimate the effect of propagating errors.
However, a combination of our method with these state-of-the-art approaches is possible.

This paper is organized as follows. 
We introduce our calibration method for syndrome measurements in Section~\ref{sec:calibration}. 
We then demonstrate how to apply it to estimate the noise including that introduced by measurement in Section~\ref{sec:noiseestimation}
and to inform a correction operation in Section~\ref{sec:depolarizing_noise}.
The Section~\ref{sec:experiment} contains a summary of results of running our method on IBM quantum computing hardware.
We conclude in Section~\ref{sec:conclusion}.

\section{Calibrated syndrome measurements}
\label{sec:calibration}

We consider the following, typical scenario for quantum error correction, see also Fig.~\ref{fig:scenario}: Syndrome measurements are performed to identify the effect of noise on encoded data and then apply appropriate corrections.
However, when performing syndrome measurements on noisy devices like the ones available today, they will introduce errors themselves.
These need to be accounted for in the correction operation.
Our work therefore aims to answer the following question:  What would be the expectation values of ideal syndrome measurements following noisy ones?
	Our proposed calibration protocol contains the following steps.
	\begin{enumerate}
		\item A calibration experiment, separate from the intended use of the syndrome measurement in the scenario described above, is carried out.
		It consists of performing the exact same syndrome measurement twice in row on the same quantum system, see Fig.~\ref{fig:calibration}.
		\item Correction factors for the expectation values are calculated from the measurement statistics of the calibration experiment.
		\item The calibration data is used to calculate ideal expectation values immediately after each occurrence of the noisy syndrome measurement in the original scenario.
	\end{enumerate}
Later on, we will give a more detailed description of steps 1 and 2 in Algorithm~\ref{alg:calibration} and step 3 in Algorithms~\ref{alg:correctionaverage} and \ref{alg:correction}.

\begin{figure}
	\begin{quantikz}
		\lstick{\calibrationinitialstate} & \qwbundle[alternate]{} &  \meter{Syndrome}\arrow[d,"\gamma_s"]\arrow[r,bend left,shorten=1pt, "\beta_s"]\qwbundle[alternate]{} &\qwbundle[alternate]{}&[1cm] \meter{Syndrome}\arrow[d,"\gamma_s"] \qwbundle[alternate]{}& \qwbundle[alternate]{}\\
		&\slice[label style={anchor=south west,rotate=40}]{$\expectationcalibrationideal$}	& \firstexpectationvalue\arrow[rr, "\gamma_s^{-1}\beta_s\gamma_s=\beta_s"]\arrow[ru,"\alpha_s",shorten=3pt] &\slice[label style={anchor=south west,rotate=40}]{$\expectationforestimation$} & \secondexpectationvalue & 
	\end{quantikz}	
	
	\caption{The shown calibration measurement is performed to characterize the noise introduced by the syndrome measurement, see Fig.~\ref{fig:scenario}, by $\alpha_s=\frac{\beta_s}{\gamma_s}$. The arrow labels indicate multiplicative factors to the expectation value. The ideal expectation value at the first dashed line is a known quantity. The first measured expectation value $\firstexpectationvalue$ allows to estimate $\gamma_s$, while $\firstexpectationvalue$ and $\secondexpectationvalue$ together allow to estimate $\beta_s$.}\label{fig:calibration}
\end{figure}

Before introducing our calibration method, let us first introduce some key concepts and notation.
More details can be found e.g. in \cite{nielsenchuang}.
We denote the Pauli matrices with $X$, $Y$, and $Z$.
The Pauli matrices with the usual multiplication generate a group $P_1$, which can be extended to $P_n$ on $n$ qubits using the tensor product. 
A Pauli operator $s\in P_n$ stabilizes a state $\ket{\psi}$ if $s \ket{\psi} = \ket{\psi}$.
The set of all stabilizers of a subspace of the state space form a subgroup of the Pauli group: the stabilizer group of this subspace.
It can be used to encode logical information distributed over several qubits.
In this context we call the subspace the code space of a stabilizer code.
The outcomes of a measurement of a generating set of the stabilizer group determine whether a state lies within the code space. 
They are called the syndrome and used to inform the correction operation. 
We denote with $[[n,k,d]]$ a stabilizer code~\cite{gottesman1997stabilizer}, where $n$ is the number of physical qubits, $k$ the number of logical qubits, and $d$ the code distance, which can correct $\lfloor \frac{d-1}{2} \rfloor$ single qubit errors.

We assume that this noise is independent of its position in a circuit. 
While this assumption is necessary for any kind of calibration, 
real devices will deviate from it to varying degrees due to cross-talk and other imperfections. For the sake of clarity, we do not include finite sampling effects here.
Furthermore, as motivated in the introduction, we assume that the noise can be described by a Pauli channel, i.e. by randomly occurring Pauli errors.
Concretely, a faulty measurement can be modeled as an ideal circuit followed by an error channel on data as well as ancilla qubits, and finally an ideal measurement in the computational basis.
We note that any circuit consisting of Clifford gates (which are gates that map Pauli operators to Pauli operators), each of which are affected by individual error channels, fits into this model. We refer the interested reader to Appendix~\ref{app:details} for a detailed mathematical exposition. More pragmatically, Eqs.~(\ref{eq:post_measurement_state}) and (\ref{eq:faulty_probability}) below may simply be taken as a definition of the behavior of faulty measurements.

We start by introducing the correction for the expectation value of stabilizer elements which is independent of the measurement outcome.
The expectation value of a stabilizer element conditioned on a given outcome will be described afterwards.

\subsection*{Calibration and correction factors for expectation values}

Consider a Pauli error $e\in P_n$ introduced into the circuit by the noisy syndrome measurement. As mentioned above, according to our noise model this error occurs between the ideal implementation of the circuit and a measurement in the computational basis.
With respect to a given stabilizer element $s$, there are two effects that the error $e$ can have:
\begin{itemize}
	\item $e$ can anti-commute with $s$ after the syndrome measurement circuit. In this case we say that $e$ propagates as an error w.r.t. $s$. That is $e$ flips any later measurement of $s$.
	\item $e$ can anti-commute with the observable of the computational basis measurement which is used to measure $s$. In this case $e$ flips the outcome of the present measurement of $s$.
\end{itemize}
Let $p_s$ and $q_s$ be the probability for the first and second effect, respectively, when drawing random Pauli errors according to the Pauli channel describing the noise.
If an anticommuting error occurs with probability $p$, then the expectation value of $s$ is damped by a factor of $1-2 p$.
Let $\beta_s$ and $\gamma_s$ be these factors associated with errors of the first and second category, respectively.

Throughout our article we use $\langle s \rangle$ and $\widetilde{\langle s \rangle}$ to denote the expectation value of a stabilizer element $s$ when calculated from ideal and noisy measurement outcomes, respectively.
From the calibration experiment, see Fig.~\ref{fig:calibration}, the expected values $\firstexpectationvalue$ and $\secondexpectationvalue$ of the first and second noisy measurement, respectively, are calculated.
By considering which errors affect which outcomes, see also Fig.~\ref{fig:calibration}, we get that $\firstexpectationvalue = \gamma_s \expectationcalibrationideal$ and $\secondexpectationvalue = \beta_s \gamma_s \expectationcalibrationideal$.
Thus, independently of the input state, the two error processes described above introduce factors 
\begin{alignat}{2}
	\beta_s:=& 1- 2 p_s = \frac{\secondexpectationvalue}{\firstexpectationvalue}
\text{ and }	\gamma_s := 1 - 2 q_s = &\frac{\firstexpectationvalue}{\expectationcalibrationideal}\label{eq:betagamma}
\end{alignat}
to the expectation value of $s$.
Indeed, the same factors apply outside the calibration experiment whenever using the same syndrome measurement. 
As mentioned above, in this article we neglect the effects of the finite sample size in an experiment
and assume for simplicity that the obtained outcome frequencies yield the expectation values up to sufficient accuracy. 

The expectation value $\expectationcalibrationideal$ of the initial state $\calibrationinitialstate$ of the calibration measurement can be calculated
because we assume that this state can be prepared perfectly. 
Evaluating Eq.~(\ref{eq:betagamma}) requires $\expectationcalibrationideal\neq 0$, 
which can be achieved, for example, by using a code word as the initial state.
Another possible choice is the product state
\begin{equation}
	\calibrationinitialstate :=  \left(\sqrt{\frac{1}{6}(3+\sqrt{3})}\ket{0} + e^{i\frac{\pi}{4}} \sqrt{\frac{1}{6}(3-\sqrt{3})}\ket{1}\right)^{\otimes n}. \label{eq:productstate}
\end{equation}
We remark that $\calibrationinitialstate$ cannot be written in the form $\cos\left(\frac{p}{q}\pi\right) \ket{0} + e^{i \frac{\pi}{4}} \sin\left(\frac{p}{q}\pi\right)\ket{1}$ with $p,q\in\mathds{Z}$, see \cite{Bogosel2018}.
Let $w=\mathrm{wt}(s)$ be the weight of a stabilizer element $s$, i.e. the number of tensor factors which are not $\1$.
Then for a stabilizer element $s$ with no phase and weight $w=\mathrm{wt}(s)$,
this choice of the initial state leads to the expectation value
\begin{equation}
	\expectationcalibrationideal = \bra{\psi} s \ket{\psi} 
	= \frac{1}{\sqrt{3}^w}.
\end{equation}
Syndrome measurements determine our subsequent correction. 
However, since the syndrome measurement itself introduces noise, this information differs before the measurement and before the correction.
We are interested in the latter, more precisely
in the expectation value $\langle s \rangle$ of an imagined, ideal stabilizer measurement just before the correction operation.
Thus, we calculate
\begin{equation}
	\alpha_s := \frac{\beta_s}{\gamma_s} \label{eq:alpha}
\end{equation}
and use the expectation values
\begin{equation}
	 \expectationforestimation = \alpha_s \expectationinsyndrome. \label{eq:correctedexpectationvalue}
\end{equation}
We emphasize again that under the given assumptions Eq.~(\ref{eq:correctedexpectationvalue}) gives the exact value of the expectation value of $s$ after the noisy measurement.
In Section~\ref{sec:examples}, we present concrete examples where the use of $\expectationforestimation$ and derived quantities is advantageous compared to $\expectationinsyndrome$.

However, here we ``forgot'' the outcome of the syndrome measurement.
For the syndrome-based error correction it will be prudent to condition on an observed syndrome outcome $x$.
In the remainder of this section we adopt the above procedure accordingly.
That is, we are going to calculate $\langle s\rangle_{\rho_x}$, where $\rho_x$ is the density matrix of the state given $x$. 

We choose generators $S_1, S_2, \ldots, S_m$, where $m=n-k$, of the stabilizer group $\thestabilizergroup=\langle S_1, S_2,\ldots, S_m\rangle$ and denote the $a$-th stabilizer element by
\begin{equation}
	S(a) \coloneqq S_1^{a_1} S_2^{a_2}\cdots S_m^{a_m}, \label{eq:Sofa}
\end{equation}
$a_i$ being the $i$-th binary digit of $a$.
By a slight abuse of notation we will also identify $a\in\{0,1,\ldots,2^m-1\}$ with the tuple of its binary digits depending on the context. Then we write $a \oplus b$ and $a \cdot b$ for the component-wise addition and the inner product, respectively, modulo two, of these binary tuples.

Let $P(e,u)$ be the joint probability for Pauli errors $e$ on the system and classical bit errors $u\in\{0,1\}^m$ on the syndrome measurement result. In the following formulas the variable $e$ will always be a Pauli operator and $a,x,u\in\{0,1\}^m$. Unless otherwise mentioned, all sums will implicitly run over the full range. We will also use the letter $P$ for the probabilities obtained by averaging out one type of error:
\begin{equation}
	P(e)\coloneqq\sum_u P(e,u)\quad\text{and}\quad	P(u)\coloneqq\sum_e P(e,u).
\end{equation}

Furthermore, $p(x|\rho)=\tr(\rho \pi_x)$ is the ideal probability of the syndrome $x$ in the state $\rho$,
where $\pi_x = \prod_i(\1+(-1)^{x_i} S_i)/2$ is the corresponding projection.
Given input $\rho$ the post-measurement state takes the form
\begin{alignat}{2}
	\rho_x =& \frac{1}{\tilde{p}(x|\rho)}\sum_{e, u} P(e,u) e\pi_{x\oplus u}\rho\pi_{x\oplus u}e^\dagger,\label{eq:post_measurement_state}
\end{alignat}
where
\begin{equation}
	\tilde{p}(x|\rho) = \sum_u P(u) p(u\oplus x\vert\rho)\label{eq:faulty_probability}
\end{equation}
is the probability of obtaining $x$ (cf.\ Proposition~\ref{prop:probability_expectedvalues}, Remark~\ref{rem:qubits}).

We can calculate the expectation value of $S(a)$ in the state $\rho_x$, even from the noisy measurement result:
\begin{alignat}{2}
	\langle S(a)\rangle_{\rho_x}
	\overset{\mathrm{Eq.}~(\ref{eq:post_measurement_state})}{=}&\frac{1}{\tilde{p}(x|\rho)}\sum_{e, u} P(e,u) \tr\left(\pi_{x\oplus u}\rho\pi_{x\oplus u}e^\dagger S(a)e\right)\nonumber\\
	=&\frac{1}{\tilde{p}(x|\rho)}\sum_u (-1)^{a\cdot(x\oplus u)} p(x\oplus u|\rho)\beta_{S(a),u}, \label{eq:expectationvaluegivenoutcome}
\end{alignat}
where in the second equality we introduced 
\begin{equation}
	\beta_{s,u}\coloneqq\sum_e P(e,u)(-1)^{(s,e)} \label{eq:betaSau_def},
\end{equation}
with
\begin{equation}	
	(e,f) =\begin{cases} 1 & \text{if $e$ and $f$ commute,}\\ 0 & \text{otherwise,} \end{cases}
\end{equation}
for two Pauli operators $e$ and $f$,
and used the fact that $(-1)^{a\cdot (x\oplus u)}$ is the outcome of a stabilizer $S(a)$ given a syndrome $x\oplus u$.
See Theorem~\ref{thm:expected_values} in the supplemental material for a more formal presentation.

All quantities in Eq.~(\ref{eq:expectationvaluegivenoutcome}) can be calculated from experimental data as follows.
The probability distribution of the noisy measurements $\tilde{p}(x|\rho)$ is directly obtained from the experiment.
The probability distribution $p(x|\rho)$ can be calculated as the Fourier transform $\mathcal{F}[.]$ of the pre-measurement expectation values, which are obtained from $\widetilde{\langle S(a)\rangle_\rho}$ with the help of $\gamma_{S(a)}$ (see also Corollary~\ref{cor:inverse_fourier}, Remark~\ref{rem:qubits}), i.e. 
\begin{equation}\label{eq:prob_via_gamma}
	p(x|\rho) = \frac{1}{2^m}\sum_a (-1)^{a\cdot x} \gamma_{S(a)}^{-1} \widetilde{\langle S(a) \rangle}_{\rho}.
\end{equation}
Finally, note that the sum in Eq.~(\ref{eq:expectationvaluegivenoutcome}) is the convolution of the functions $(-1)^{a\cdot u}p(u|\rho)$ and $\beta_{S(a),u}$. Applying a Fourier transform gives the pointwise product of the Fourier transforms of these functions (convolution theorem), such that we can solve for $\beta_{S(a),u}$, see also Theorem~\ref{thm:expected_values}, Remark~\ref{rem:qubits}. 
This allows to express $\beta_{S(a),u}$ via 
\begin{alignat}{2}
	\beta_{S(a),u}=&\mathcal{F}^{-1}\left[\frac{\mathcal{F}[\langle S(a) \rangle_{\rho_x} \tilde{p}(x|\rho)](y)}{\mathcal{F}[(-1)^{a\cdot x} p(x|\rho)](y)}\right](u)\nonumber\\
	=&\mathcal{F}^{-1}\left[\frac{\sum_x (-1)^{x\cdot y} \langle S(a) \rangle_{\rho_x} \tilde{p}(x|\rho)}{\langle S(a+y) \rangle}\right](u)\nonumber\\
	=&\frac{1}{2^m}\sum_{b,x} (-1)^{(u\oplus x)\cdot b}\frac{ \langle S(a) \rangle_{\rho_x} \tilde{p}(x|\rho)}{\langle S(a+y) \rangle}.
\intertext{Inserting the state of the calibration experiment and using $\gamma_{S(a)}$ to replace the ideal expectation value by a noisy one yields}
		\beta_{S(a),u}=&\frac{1}{2^m}\sum_{b,x}(-1)^{(u\oplus x)\cdot b}\frac{\widetilde{\langle S(a)\rangle}_{\mathrm{cal.},x}^{(2)} \tilde{p}(x\vert\rho_{\mathrm{cal.}})}{\gamma_{S(a)} \langle S(a\oplus b) \rangle_{\mathrm{cal.}} }, \label{eq:betaSau_cal}
\end{alignat}
where the index $x$ in $\widetilde{\langle S(a)\rangle}_{\mathrm{cal.},x}^{(2)}$ indicates that this expectation value is calculated for a given outcome $x$ in the first measurement.
Eq.~(\ref{eq:betaSau_cal}) allows to obtain $\beta_{S(a),u}$ from measured data and known quantities only.

We give a summary of the calibration and correction in Appendix~\ref{app:algos}, Algorithm~\ref{alg:calibration}, Algorithm~\ref{alg:correctionaverage}, and Algorithm~\ref{alg:correction}, respectively.


To conclude, we connect back to the initial calibration description and revisit Eq.~(\ref{eq:betagamma}) from the point of view 
that Eqs.~(\ref{eq:post_measurement_state}) and (\ref{eq:faulty_probability}) describe the behavior of faulty measurements.
The rates used in Eq.~(\ref{eq:betagamma}) can be expressed in terms of the probability distribution $P$ as
\begin{equation}
	p_s=\sum_{(e,s)=1}P(e)\quad\text{and}\quad q_s=\sum_{(S(u),s)=1}P(u).
\end{equation}
Representing the measurement result of the stabilizer $S(a)$ by means of the random variable
\begin{equation}
	\chi_a\colon\{0,1\}^m\rightarrow \{\pm 1\},\quad\chi_a(x)\coloneqq(-1)^{a\cdot x}=(-1)^{(S(a),S(x))},
\end{equation}
and taking the expectation value with respect to the probability distribution in Eq.~(\ref{eq:faulty_probability})
yields the formula
\begin{equation}\label{eq:gamma_factor}
	\widetilde{\langle S(a) \rangle}_{\rho}=\left(\sum_{u}(-1)^{a\cdot u}P(u)\right)\langle S(a) \rangle_{\rho}.
\end{equation}
While this can be seen by direct computation, conceptually this also follows because the right-hand side in (\ref{eq:faulty_probability}) is a convolution (cf.\ Proposition~\ref{prop:probability_expectedvalues}).
Recognizing the multiplicative factor as $\gamma_{S(a)}=1-2q_{S(a)}$, we see that this is consistent with Eq.~(\ref{eq:betagamma}). The average over all outcomes of the post-measurement states can be written as
\begin{equation}\label{eq:post_measurement_avg}
	\rho_{\avg}\coloneqq\sum_x \tilde p(x\vert\rho)\rho_x =\sum_eP(e) e\left(\sum_x \pi_x\rho\pi_x\right)e^\dagger.
\end{equation}
Hence the error channel of propagating errors is the composition of a Pauli channel with the channel described by the projection operators. We note that the latter does not affect stabilizer measurements. From Eq.~(\ref{eq:post_measurement_avg}) we obtain the formula
\begin{equation}\label{eq:avgbeta}
	\langle s \rangle_{\rho_{\avg}}=\left(\sum_e(-1)^{(s,e)}P(e)\right) \langle s \rangle_{\rho},
\end{equation}
where the multiplicative factor can again be recognized as $\beta_s=1-2p_s$ in agreement with Eq.~(\ref{eq:betagamma}).
We also remark that this is recovered by averaging Eq.~(\ref{eq:expectationvaluegivenoutcome}) over all $x\in\{0,1\}^m$ and observing the equality
\(
\sum_x\beta_{s,x}=\beta_s
\).

\section{Examples}\label{sec:examples}
In this section we illustrate our method by calculating an example: 
The $[[7,1,3]]$ Steane code~\cite{Steane96b}, with two different noise models.

Instead of the more common sequential syndrome measurement,
we use the parallelized circuit shown in Appendix~\ref{app:Steane}, Fig.~\ref{fig:flagqubitcircuits}.
It is taken from \cite{Reichardt2020}, where it appears in the context of flag qubits.
We do not use this circuit for its fault tolerance properties,
but merely to illustrate the generality of our method.
Apart from our assumptions described above, mainly the assumption of Pauli noise,
at no point during the derivation of our calibration procedure did we need to specify the exact inner workings of the syndrome measurement. We implicitly use that faulty measurements can be formally composed (cf.\ Proposition~\ref{prop:composition_measurements}).

For the examples we use the following simple noise model to be able to plot the results easily.
Our proposed method is not restricted to this kind of noise, see above.
We assume single qubit gates (only the Hadamard gate $H$ in our case) are perfect.
Each two-qubit gate acts ideally, 
followed by a Pauli channel on the same qubits.
We denote the noisy version of the unitary gate $U_{ij}$ acting on qubits $i$ and $j$ with $\tilde{U}_{ij}$.
It is the quantum channel
\begin{equation}
	\tilde{U}_{ij}(p, \rho) = \epsilon_p^{(i,j)} (U_{ij}\rho U_{ij}^\dagger),
\end{equation}
where
\begin{equation}
	\epsilon_p^{(i,j)}(\rho):= \sum_{u,v=0}^3 p_{uv} \sigma_{u}^{(i)}\otimes \sigma_{v}^{(j)} \rho \sigma_{u}^{\dagger(i)}\otimes \sigma_{v}^{\dagger(j)},
\end{equation}
is a general two-qubit Pauli channel acting on qubits i and j with the probability $p_{uv}$ for the error $\sigma_{u}^{(i)}\otimes \sigma_{v}^{(j)}$. 
Later gates propagate errors introduced by earlier ones. 
This effect is therefore included in the calculations of measurement statistics below. 

\subsection{Estimation of the noise channel}
\label{sec:noiseestimation}
This section is concerned with estimating the noise channel $\cE$ (as in Fig.~\ref{fig:scenario}) combined with the propagating errors introduced by the measurement (cf.\ Eq.~\ref{eq:post_measurement_avg}).

In general, there is a procedure described by Thomas Wagner et.\ al.\ in \cite{wagner2022pauli} to reconstruct a Pauli channel from stabilizer measurements under the assumption that it is built out of independent ones acting in a sufficiently local manner.
In our specific example of the seven-qubit-Steane code this amounts to the following problem, which can be solved by direct computation (see Appendix~\ref{app:Steane}):

Given a Pauli channel consisting of independent channels on the individual qubits, i.e.\ the probability mass function $P\colon \mathbb P_7\rightarrow [0,1]$ is of the form
\begin{equation}
	P(e_1\otimes \cdots\otimes e_7)=P_1(e_1)\cdots P_7(e_7),\quad e_i\in \{\1, X, Y, Z\},
\end{equation}
determine $P$ from the expectation values of stabilizer measurements. 

In the spirit of this paper we ask how one should calculate the expectation values used for this channel estimation and compare the following two approaches:
\begin{enumerate}
	\item Calculate the expectation value of a stabilizer element $s$ directly from the noisy measurement results.
	\item Calculate the expectation value via our calibration method, 
	see Eq.~(\ref{eq:correctedexpectationvalue}).
\end{enumerate}
We do this for the explicit example of depolarizing noise in the measurement circuit. In order to focus on the effects due to the measurement we set $\cE=\Id$.
For the two-qubit depolarizing channel 
\begin{equation}
	\Delta_\lambda (\rho) = (1-\lambda) \rho + \lambda \frac{\1}{4} \label{eq:depolarizingnoise}
\end{equation}
with $\lambda\in[0,1]$, 
the error probabilities are
\begin{equation}
	p_{uv}=\left\{\begin{array}{cl}
		1-\frac{15}{16}\lambda & \text{if } u=v=0\\
		\frac{1}{16}\lambda & \text{else.}
	\end{array}\right.
\end{equation}

First, we look at the calibration measurement.
We use the state $\ket{\psi}$ of Eq.~(\ref{eq:productstate}), which has
\begin{equation}
	p(x|\proj{\psi}) = \left\{\begin{array}{cl}
		\frac{1}{108} & \text{if }\mathrm{wt}(S(x))=6\\
		\frac{1}{36} & \text{else,}
	\end{array}\right.
\end{equation}
as can be confirmed by a direct calculation. From the expectation values of the first measurement in the calibration experiment we can calculate the factors $\gamma_s$, see Eq.~(\ref{eq:betagamma}).
We inspect the circuit, see Fig.~\ref{fig:flagqubitcircuits}, and consider pairs of gates and measurements.
From this we see that for our simple noise model there are only three possibilities: the gate noise does not affect the measurement, only $X$- or only $Z$-errors affect the measurement. 
Therefore the factors $\gamma_s$ only depend on the number $n_{\gamma,s}$ of gates which can introduce errors affecting the respective outcome.
The relation reads
\begin{equation}
	\gamma_{S(a)} = (1-\lambda)^{n_{\gamma,a}}.
\end{equation}
For each of the 64 stabilizer elements the values are given by
\begin{equation} 
	\begin{aligned}
	n_{\gamma,a}=& (0, 21, 8, 22, 10, 22, 12, 23, 11, 23, 13, 24, 13, 25, 14, 25, 20, \\
	&27, 21, 27, 22, 26, 21, 26, 19, 26, 20, 25, 22, 27, 22, 27, 21, 27, \\
	&24, 27, 23, 27, 23, 26, 20, 26, 24, 26, 21, 26, 24, 27, 19, 26, 23, \\
	&26, 25, 27, 24, 28, 24, 26, 25, 27, 26, 28, 26, 28)_a.
	\end{aligned} \label{eq:ngammaa}
\end{equation}
Analogously
\begin{equation}
	\begin{aligned}
n_{\beta,a}=&	(0, 18, 21, 16, 23, 16, 15, 15, 18, 20, 26, 26, 27, 25, 23, 23, 21, \\
		&26, 24, 28, 27, 25, 26, 25, 16, 26, 28, 18, 25, 23, 22, 24, 23, 27, \\
		&27, 25, 26, 28, 28, 26, 16, 25, 25, 23, 28, 20, 22, 23, 15, 23, 26, \\
		&22, 28, 22, 18, 22, 15, 23, 25, 24, 26, 23, 22, 18)_a,
	\end{aligned}  \label{eq:nbetaa}
\end{equation}
which allows us to calculate $\alpha_s$ according to Eq.~(\ref{eq:alpha}). 
Note that the order in Eqs. (\ref{eq:ngammaa}) and (\ref{eq:nbetaa}) is given by Eq.~(\ref{eq:Sofa}).
We plot $\alpha(\lambda)$ for some stabilizer elements in Fig.~\ref{fig:alphaoflambda}.
\begin{figure*}
	\begin{center}
		\input{images/lambdaalpha.tex}
	\end{center}
	\caption{The correction factor $\alpha_s$, see Eq.~(\ref{eq:alpha}) from our calibration for the syndrome measurement shown in Fig.~\ref{fig:flagqubitcircuits} with simple depolarizing noise, see Eq.~(\ref{eq:depolarizingnoise}). We show only the values for the generators of the stabilizer, but it can be calculated analogously for all other elements.}
	\label{fig:alphaoflambda}
\end{figure*}

We can now choose whether to use these correction factors when calculating the expectation value of any stabilizer element after the syndrome measurement.
These are then two possible inputs for the reconstructed probability distribution (see Appendix~\ref{app:Steane}, Eq.~(\ref{eq:reconstructedP})) which in turn give rise to different estimates of the noise channel. We note that due to the fact that the ideal expectation values for a stabilizer state are 1, these inputs are actually just the factors $\beta_s$ respectively $\gamma_s$ from the calibration. However, the general procedure outlined in this section does not depend on this \textit{a priori} knowledge.

To compare these two distributions to the actual error distribution, we use the Kullback-Leibler divergence~\cite{KullbackLeibler}
\begin{equation}
	D_{KL} (P_1, P_2) := \sum_{x\in\mathcal{X}}
	P_1(x) \log_2 \left(\frac{P_1(x)}{P_2(x)}\right), \label{eq:DKL}
\end{equation}
where $P_2$ plays the role of the reference.
The result is shown in Fig.~\ref{fig:DKLofLambda}. The Bhattacharyya distance~\cite{Bhattacharyya46} shows the same qualitative behavior.
\begin{figure}[thp]
	\centering
	\input{images/DKLofLambda.tex}
	\caption{The distance (see Eq.~(\ref{eq:DKL})) of the estimated and the ideal error distribution, with (orange) and without (blue) doing our calibration.}
	\label{fig:DKLofLambda}
\end{figure}

As expected, the estimated noise channel is closer to the actual one when using our calibration method.

\FloatBarrier
\subsection{Error correction with calibrated syndrome measurements}
\label{sec:depolarizing_noise}

In this section we switch to a different application of our calibration method:
We use it to improve syndrome decoding.
Typically, the error correction operation is directly based on the measured syndrome.
For the seven-qubit-Steane code a possible choice of correction operations is shown in Table~\ref{tab:error_representative}. 
\begin{table}[thp]
	\caption{Minimal weight representatives for the classes of errors with the same syndrome. The measurement outcomes are ordered as $S_1, S_2, \ldots, S_6$.}\label{tab:error_representative}
	\begin{tabular}[t]{cccccc|l}
		\multicolumn{6}{c}{Syndrome} & error\\
		\hline
 + & + & + & + & + & + & $1$\\
+ & + & + & + & + & - & $X_1$\\
+ & + & + & + & - & + & $X_2$\\
+ & + & + & + & - & - & $X_3$\\
+ & + & + & - & + & + & $X_4$\\
+ & + & + & - & + & - & $X_5$\\
+ & + & + & - & - & + & $X_6$\\
+ & + & + & - & - & - & $X_7$\\
+ & + & - & + & + & + & $Z_1$\\
+ & + & - & + & + & - & $X_1Z_1$\\
+ & + & - & + & - & + & $X_2Z_1$\\
+ & + & - & + & - & - & $X_3Z_1$\\
+ & + & - & - & + & + & $X_4Z_1$\\
+ & + & - & - & + & - & $X_5Z_1$\\
+ & + & - & - & - & + & $X_6Z_1$\\
+ & + & - & - & - & - & $X_7Z_1$
	\end{tabular}\hfill
	\begin{tabular}[t]{cccccc|l}
		\multicolumn{6}{c}{Syndrome} & error\\
		\hline
+ & - & + & + & + & + & $Z_2$\\
+ & - & + & + & + & - & $X_1Z_2$\\
+ & - & + & + & - & + & $X_2Z_2$\\
+ & - & + & + & - & - & $X_3Z_2$\\
+ & - & + & - & + & + & $X_4Z_2$\\
+ & - & + & - & + & - & $X_5Z_2$\\
+ & - & + & - & - & + & $X_6Z_2$\\
+ & - & + & - & - & - & $X_7Z_2$\\
+ & - & - & + & + & + & $Z_3$\\
+ & - & - & + & + & - & $X_1Z_3$\\
+ & - & - & + & - & + & $X_2Z_3$\\
+ & - & - & + & - & - & $X_3Z_3$\\
+ & - & - & - & + & + & $X_4Z_3$\\
+ & - & - & - & + & - & $X_5Z_3$\\
+ & - & - & - & - & + & $X_6Z_3$\\
+ & - & - & - & - & - & $X_7Z_3$
	\end{tabular}\hfill
	\begin{tabular}[t]{cccccc|l}
		\multicolumn{6}{c}{Syndrome} & error\\
		\hline
- & + & + & + & + & + & $Z_4$\\
- & + & + & + & + & - & $X_1Z_4$\\
- & + & + & + & - & + & $X_2Z_4$\\
- & + & + & + & - & - & $X_3Z_4$\\
- & + & + & - & + & + & $X_4Z_4$\\
- & + & + & - & + & - & $X_5Z_4$\\
- & + & + & - & - & + & $X_6Z_4$\\
- & + & + & - & - & - & $X_7Z_4$\\
- & + & - & + & + & + & $Z_5$\\
- & + & - & + & + & - & $X_1Z_5$\\
- & + & - & + & - & + & $X_2Z_5$\\
- & + & - & + & - & - & $X_3Z_5$\\
- & + & - & - & + & + & $X_4Z_5$\\
- & + & - & - & + & - & $X_5Z_5$\\
- & + & - & - & - & + & $X_6Z_5$\\
- & + & - & - & - & - & $X_7Z_5$
	\end{tabular}\hfill
	\begin{tabular}[t]{cccccc|l}
		\multicolumn{6}{c}{Syndrome} & error\\
		\hline
- & - & + & + & + & + & $Z_6$\\
- & - & + & + & + & - & $X_1Z_6$\\
- & - & + & + & - & + & $X_2Z_6$\\
- & - & + & + & - & - & $X_3Z_6$\\
- & - & + & - & + & + & $X_4Z_6$\\
- & - & + & - & + & - & $X_5Z_6$\\
- & - & + & - & - & + & $X_6Z_6$\\
- & - & + & - & - & - & $X_7Z_6$\\
- & - & - & + & + & + & $Z_7$\\
- & - & - & + & + & - & $X_1Z_7$\\
- & - & - & + & - & + & $X_2Z_7$\\
- & - & - & + & - & - & $X_3Z_7$\\
- & - & - & - & + & + & $X_4Z_7$\\
- & - & - & - & + & - & $X_5Z_7$\\
- & - & - & - & - & + & $X_6Z_7$\\
- & - & - & - & - & - & $X_7Z_7$
	\end{tabular}
\end{table}

The noise introduced by the measurement is not symmetric in the sense that the amount of propagated noise might differ drastically from the noise on the measurement outcomes.
When the error correction operation is chosen according to the measured syndrome only, this is not taken into account.
With the calibration, however, we can base the correction operation on the most likely result of an ideal syndrome measurement on the system after the noise was introduced.

We recall the general situation of Fig.~\ref{fig:scenario}. A stabilizer state is affected by an error channel $\cE$ yielding the state $\rho$ and we perform a measurement recording the syndrome $x\in \ftwo^m$.
In principle the following data is available from the calibration:
The values
\begin{equation}\label{eq:syndromebetas}
	\beta_{S(a),u} = \sum_e P(e,u)(-1)^{(S(a),e)} = \sum_y\left(\sum_{\Syn(e)=y}P(e,u)\right)(-1)^{a\cdot y}
\end{equation}
for all $a\in \ftwo^m$, where 
\begin{equation}
	\Syn(e)=((e,S_1),(e,S_2),\ldots, (e,S_m))\in \ftwo^m
\end{equation}
is the \emph{syndrome} of $e$ determined by the relation
\begin{equation}
	\Syn(e)\cdot a = (e,S(a)).
\end{equation}
From these the conditional probabilities (for non-vanishing $P(x)$)
\begin{equation}\label{eq:syndromebetas2}
	P_{\Syn,x}(y)\coloneqq\frac{1}{P(x)}\sum_{\Syn(e)=y}P(e,x)
\end{equation}
of all syndrome classes can be computed via the inverse Fourier transformation of Eq.~(\ref{eq:syndromebetas}) (see also Remark~\ref{rem:qubits}).

In order to offer a sensible description of maximum-likelihood decoding, the Pauli channel $\cE$ should be considered.
If it is trivial, then $p_x(y)=P_{\Syn,x}(y)$
can be interpreted as the conditional probability that an error with syndrome $y$ has occurred given that $x$ was observed (see Eq.~(\ref{eq:conditional_postmeasuremnt_stab})).
We refer to Remark~\ref{ex:maxlike} for more details in the general situation.

In this context we describe four different decoding schemes, producing the syndrome class of a potential error correction operation.
Within this class the minimal weight representative is then chosen to be applied to the post-measurement state $\rho_x$.

\begin{enumerate}[label=(D\arabic*)]
	\item\label{d1} We ignore the internal noise and take $x$ itself. This amounts to ordinary minimal-weight decoding.
	
	\item\label{d2} From the calibration data we determine the most likely syndrome class with respect to the probabilities $P(e)$ of propagating errors. This is added to the measured syndrome $x$.
	
	\item\label{d3} Another option is to select the syndrome $y$ based on the signs of the expectation values for the chosen stabilizer generators, so that they are all positive after correction. This of course assumes that the error rates are small enough to avoid vanishing values. 
	In more detail, we can use equations \ref{eq:prob_via_gamma} and \ref{eq:expectationvaluegivenoutcome} to calculate $\langle S_i\rangle_{\rho_x}$ and hence $y\in \ftwo^m$ such that $\operatorname{sign}\langle S_i\rangle_{\rho_x}=(-1)^{y_i}$.

	\item\label{d4} Given the measured syndrome $x$ we take the most likely error-syndrome $\argmax_{y} p_x(y)$ as described above.
\end{enumerate}

The noise introduced by the syndrome measurement is added to that of the channel $\varepsilon$, see Fig.~\ref{fig:scenario}. 
Specifically, the propagating Pauli errors are multiplied and the syndromes are added modulo two.
How the combined error of the two independent noise sources is handled depends heavily on the decoder.
Therefore, we want to separate their effects, and since the effect of $\varepsilon$ is well studied, we focus on the effects due to the noisy measurement in the numerical analysis by setting the error introduced by $\varepsilon$ to $\1$.
In other words, we assume that the input state $\rho$ is a code word.
Then the post-measurement state takes the form
\begin{equation}\label{eq:conditional_postmeasuremnt_stab}
	\rho_x = \sum_e P_x(e)e\rho e^\dagger,
\end{equation}
the result of applying the Pauli channel associated with the conditional distribution $P_x(e)=P(e,x)/P(x)$.
In view of the notion that a correction is successful if the actual error and the guessed one only differ by a stabilizer, the success rate of a decoding scheme that outputs the error $e(x)$ is given by
\begin{equation}
	p_{\operatorname{succ}}\coloneqq\sum_{x,a} P(x)P_x(e(x)S(a))=\sum_{x,a} P(e(x)S(a),x).
\end{equation}
We analyzed the failure rate of the decoders (D1) to (D4) for the depolarizing noise model and the 7-qubit-Steane code with the syndrome measurement circuit shown in Fig.~\ref{fig:flagqubitcircuits}.
The results are plotted in Fig.~\ref{fig:decodingstrategies}.
Using maximum likelihood decoding based on the calibration data leads to the lowest logical error rate.
We see that the absolute difference in terms of the logical error rate w.r.t. not using the calibration data can be up to 8 \% for the 7-qubit-Steane code.
Given the low resource overhead of our calibration, this is a significant improvement.

\begin{figure}[tph]
	\centering
	\input{images/decodingstrategies.tex}
	\caption{
		 The failure rates of the different decoding strategies for the 7-qubit-Steane code with the syndrome measurement shown in Fig.~\ref{fig:flagqubitcircuits} and depolarizing noise, see Eq.~(\ref{eq:depolarizingnoise}). 
		 While the naive and approximate curves are indistinguishable here, the values can differ up to roughly $0.002$ for $\lambda = 0.2$ with the approximate decoding performing better. 
		The best performance is achieved by maximum likelihood decoding based on the input from our calibration.}
	\label{fig:decodingstrategies}
\end{figure}

One can also consider examples with more drastic differences between naive and maximum likelihood decoding.
We choose the noise model where with probability $\lambda$ a $Z$-error is introduced on the control qubit.
This corresponds to
\begin{equation}
	p_{00}= 1-\lambda \text{ and } p_{30}= \lambda 
\end{equation}
and all other $p_{uv}=0$.
To demonstrate that the effect does not only occur for the 7-qubit Steane code that we have mainly considered, we will also examine the 5-qubit perfect code with the circuit shown in Figure~\ref{fig:fivequbitshortype}.
For both circuits, we assume for each two-qubit gate that the control qubit is the lower of the two qubits.
Direct calculation yields the probability of no flip on the measurement outcomes
\begin{equation}
\begin{aligned}
	P_7(u=0) =& (\lambda -1)^4 \left(8 \lambda ^4-16 \lambda ^3+12 \lambda ^2-4 \lambda +1\right)^2 \left(16 \lambda ^4-24 \lambda ^3+16 \lambda ^2-4 \lambda +1\right)^4\\
\text{and }	P_5(u=0) =& (4 (\lambda-1) \lambda (2 (\lambda-1) \lambda+1)+1)^4,
\end{aligned}
\end{equation}
for the 7-qubit-Steane code and the 5-qubit perfect code, respectively.

For the sake of clarity let us assume that the input state is a code word.
Additional correctable errors on the input state do not change the following argument, but they unnecessarily obfuscate it.
Because with the described noise model the errors are only introduced on the ancillary qubits, there is no need to perform a correction operation.
And because our method yields the correct expectation values $\langle S(a)\rangle_{\rho_x} = 1$, we never apply a correction operation, as needed.
However, with probability $1-P(u=0)$ the syndrome differs from the all $+1$ outcome and the standard approach fails to identify the appropriate correction operation.
For this extreme example the difference in the two approaches is very pronounced, as can be seen from Fig.~\ref{fig:probabilityofdisagreement}. 
Indeed, for both codes there are parameter ranges of the noise parameter where the standard syndrome decoding fails almost certainly, while with calibration it never fails in this scenario.
The behavior when the noise parameter is $\lambda = 1$, i.e. when errors are introduced with certainty, depends on whether the measured qubits act as a control for an odd or even number of gates.	

Note that while the difference will not be this big in other scenarios, our correction can only lower the logical error rate, given the assumptions made in our derivation are fulfilled.
\begin{figure}[tph]
	\begin{quantikz}[row sep=1pt, column sep=1pt]
		& \targ{} 	& \qw 		& \qw 		& \qw 		& \qw & \qw & \qw 
		& \qw 		& \qw 		& \qw 		& \qw 		& \qw & \qw & \qw
		& \qw		& \qw 		& \qw 		& \targ{}	& \qw & \qw & \qw
		& \qw		& \qw 		& \ctrl{}	& \qw		& \qw & \qw & \qw		
		\\
		& \qw 		& \ctrl{} 	& \qw		&\qw 		& \qw & \qw & \qw
		& \targ{}	& \qw 		& \qw 		&\qw 		& \qw & \qw & \qw
		& \qw 		& \qw 		& \qw 		&\qw 		& \qw & \qw & \qw
		& \qw 		& \qw 		& \qw 		&\targ{}	& \qw & \qw & \qw
		\\
		& \qw 		& \qw  		& \ctrl{} 	&\qw 		& \qw & \qw & \qw
		& \qw 		& \ctrl{} 	& \qw 		&\qw 		& \qw & \qw & \qw
		& \targ{}	& \qw 		& \qw 		&\qw 		& \qw & \qw & \qw
		& \qw		& \qw 		& \qw 		&\qw 		& \qw & \qw & \qw
		\\
		& \qw 		& \qw 		& \qw 		& \targ{} 	& \qw & \qw & \qw
		& \qw 		& \qw 		& \ctrl{}	& \qw 		& \qw & \qw & \qw
		& \qw 		& \ctrl{}	& \qw		& \qw 		& \qw & \qw & \qw
		& \targ{}	& \qw		& \qw		& \qw 		& \qw & \qw & \qw
		\\
		& \qw 		& \qw 		& \qw 		& \qw 		& \qw & \qw & \qw 
		& \qw		& \qw 		& \qw 		& \targ{}	& \qw & \qw & \qw
		& \qw		& \qw 		& \ctrl{}	& \qw		& \qw & \qw & \qw
		& \qw		& \ctrl{}	& \qw		& \qw		& \qw & \qw & \qw
		\\
		\lstick[4]{$\ket{\mathrm{GHZ}_4}$} & \ctrl{-5} & \qw & \qw & \qw & \meter{X} 
		&\hspace{1.5cm} & 
		\lstick[4]{$\ket{\mathrm{GHZ}_4}$} & \ctrl{-4} & \qw & \qw & \qw & \meter{X}
		&\hspace{1.5cm}& 
		\lstick[4]{$\ket{\mathrm{GHZ}_4}$} & \ctrl{-3} & \qw & \qw & \qw & \meter{X}
		&\hspace{1.5cm}& 
		\lstick[4]{$\ket{\mathrm{GHZ}_4}$} & \ctrl{-2} & \qw & \qw & \qw & \meter{X}
		\\
		& \qw & \ctrl{-5} & \qw & \qw & \meter{X}&&
		& \qw & \ctrl{-4} & \qw & \qw & \meter{X}&&
		& \qw & \ctrl{-3} & \qw & \qw & \meter{X}&&
		& \qw & \ctrl{-2} & \qw & \qw & \meter{X}&&
		\\
		& \qw & \qw & \ctrl{-5} & \qw & \meter{X}&&
		& \qw & \qw & \ctrl{-4} & \qw & \meter{X}&&
		& \qw & \qw & \ctrl{-3} & \qw & \meter{X}&&
		& \qw & \qw & \ctrl{-7} & \qw & \meter{X}&&
		\\
		& \qw & \qw & \qw & \ctrl{-5} & \meter{X}&&
		& \qw & \qw & \qw & \ctrl{-4} & \meter{X}&&
		& \qw & \qw & \qw & \ctrl{-8} & \meter{X}&&
		& \qw & \qw & \qw & \ctrl{-7} & \meter{X}&&
	\end{quantikz}
	\caption{A Shor-type syndrome measurement~\cite{Shor1996} for the Five-qubit code. It uses a source of $\ket{\mathrm{GHZ}_4}=(\ket{0000}+\ket{1111})/\sqrt{2}$ states, which we assume to be perfect. The parity of each block of four parallel $X$-basis measurements determines one bit of the syndrome corresponding to one stabilizer generator.}\label{fig:fivequbitshortype}
\end{figure}
\begin{figure}[thp]
	\centering
	\def\svgwidth{0.75\textwidth}
	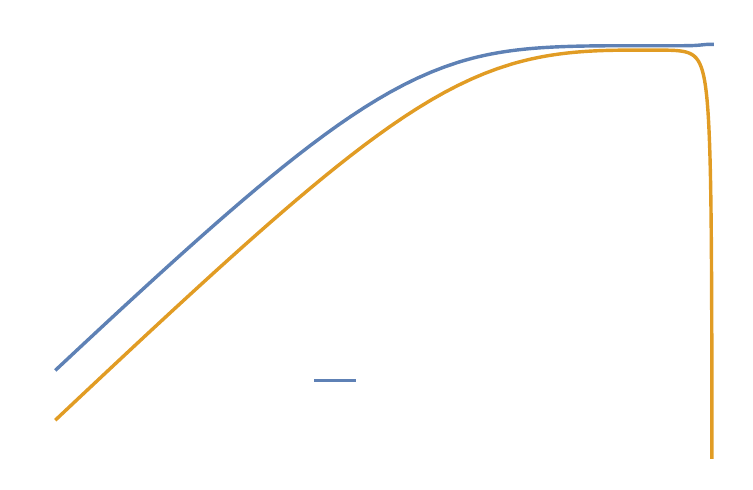
	\caption{The probability of a wrong syndrome outcome as a function of the noise parameter $\lambda$. With this probability the standard syndrome correction leads to logical errors, while the calibration allows to always do the appropriate correction in this scenario, see Section~\ref{sec:depolarizing_noise}.}
	\label{fig:probabilityofdisagreement}
\end{figure}

\section{Experiment}
\label{sec:experiment}
In this section we evaluate the performance of our method on a real quantum computer.
The main challenge is the sparse connectivity of currently available backends, such as those provided by the IBM Quantum Experience~\cite{IBMQuantum}. 
As a result, circuits must be compiled in such a way as to be executable using physically realizable gate operations, in particular,
carrying out two-qubit operations only between qubits that couple to one another.

In general, this entails that information stored in a qubit has to be transported along the backend using swap-operations. However, this presents a problem: as our calibration method necessitates carrying out the same measurement twice, in the compiled circuit, a first measurement will in general change the association between circuit- and backend-qubits. Thus, the physical instantiation of a second stabilizer measurement will no longer be the same, and we have no reason to suppose that noise introduced during this measurement should equal that introduced by the first.

To meet these challenges, stabilizer measurements have to be defined in such a way as to both respect the backend topology, and leave the association between circuit- and backend-qubits (the `embedding map') invariant. 
To ensure that our experiment satisfies this constraint, we manually compiled the circuit, see Appendix~\ref{app:hardwaregraph} for details. 
Concretely, we implement the 7-qubit Steane code on backends provided by the IBM Quantum Experience.
However, the adaptations required to implement these measurements on the physical backend significantly increase the circuit depth: to implement a single CNOT operation in the syndrome measurement, a total of 13 CNOTs must be physically performed. 
This means that a simultaneous measurement of all stabilizer operators on the quantum backend is prohibitively expensive, given the capabilities of current NISQ systems.

Nevertheless, it is possible to demonstrate our method by means of running the calibration (i.e. measuring the same stabilizer twice), then using the correction factors to estimate the expectation value after the noisy measurement. 
To keep running the experiment on currently available NISQ-devices feasible, given the large overhead induced by mapping the stabilizer measurements to the device topology, in the following, we will concentrate on evaluating our method for the measurement of individual stabilizer generators.
Thus, for each experiment, we will measure the same generator twice for the calibration, then use the correction factors estimated from this data to obtain an improved guess at the post-measurement value in an experiment involving only the measurement of that generator.

\subsection{Noisy Simulation}
\label{subsec:noise}

To first find the effect of implementing our method using the operations schedule discussed in the preceding section, we simulate the execution of the calibration and a single measurement for a given stabilizer generator. 
We use the simulator provided by Qiskit~\cite{qiskit}.

For the stabilizer generators involving only `Z'-measurements, the system was initialized in the state $\ket{00\ldots 0}$, such that the ideal expectation value is equal to $1$. We use different noise models to illustrate the method. In Fig.~\ref{fig:sim_bitflip}, every operation, including qubit reset and measurement, carries a $1\%$ chance of flipping any bit being operated upon. Fig.~\ref{fig:sim_meas} shows a scenario in which only measurements carry a $5\%$ chance of flipping the outcome, thus, no errors are propagated by the measurement. The ideal expectation value is recovered to a degree only limited by the shot noise inherent in estimating the correction factor (using in each case $10^4$ shots).

\begin{figure}[tp]
	\centering
	\subfloat[ 	\label{fig:sim_bitflip}]{
        \begin{overpic}[width=0.49\linewidth]{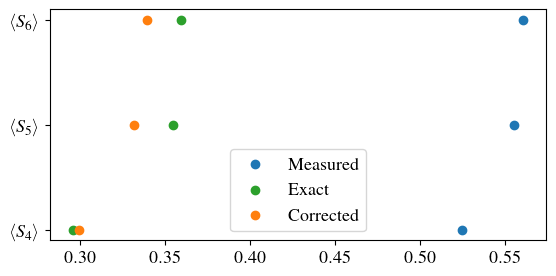}
			
	\end{overpic}}\hfill
	\subfloat[ \label{fig:sim_meas}]{	
        \begin{overpic}[width=0.49\linewidth]{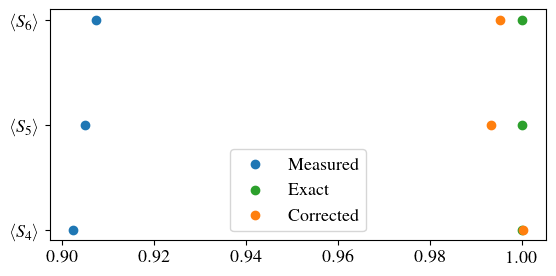}
			
	\end{overpic}}
	\caption{Measurement of the three stabilizer generators $S_4$, $S_5$, and $S_6$, subject to (a) bit flipping noise and (b) measurement noise. 
	Our method maps the measured expectation values (blue) to corrected ones (orange) and recovers the true values (green).
	The remaining gap is only due to the finite number of experimental runs ($10^4$ shots).}
	\label{fig:sim}
\end{figure}

The data for the exact values in Fig.~\ref{fig:sim} was obtained by using the Qiskit simulator's ability to take snapshots of the exact state during the computation. 
However, in a scenario involving a real device, this is obviously not feasible. 

To nevertheless obtain a point of comparison, we follow the measurement of a stabilizer generator with single qubit measurements on all data qubits, from which the value after the stabilizer measurement can be computed. 
If these `final' qubit measurements can be idealized as noiseless, we would obtain an exact value that could directly be compared to the value reconstructed by means of the correction factor.
However, as these measurements are themselves noisy, further errors introduced have to be corrected. 
For a low number of qubits $n$, such as in this case, it is feasible to perform full Readout Error Mitigation (REM) (see~\cite{cai2023quantum} and references therein): 
prepare and measure all $2^n$ basis states in the computational basis, and then use the obtained results to compute a transition matrix whose inverse allows to correct the observed distribution.

With this, we obtain four values per measured stabilizer generator: the outcome of the stabilizer measurement itself, the value corrected using our method, the unmitigated result from the final single qubit measurements, and the result after REM. 

\begin{figure}[tp]
	\centering
	\subfloat[ 	\label{fig:REM_cairo}]{
        \begin{overpic}[width=0.49\linewidth]{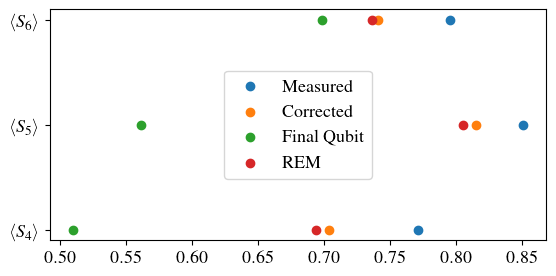}
			
	\end{overpic}}\hfill
	\subfloat[ \label{fig:REM_hanoi}]{	
        \begin{overpic}[width=0.49\linewidth]{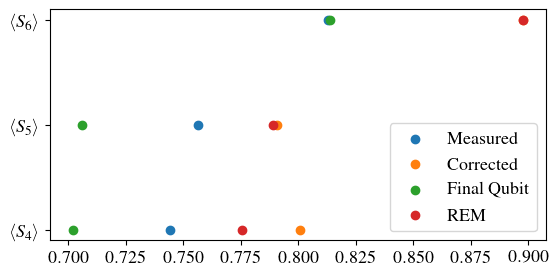}
			
	\end{overpic}}
	\caption{Measurement of the three stabilizer generators $S_4$, $S_5$, and $S_6$, as carried out on (a) the simulated ibm\_cairo-device, (b) the simulated ibm\_hanoi-device. The measured expectation values (blue) are corrected with our proposed method (orange). As a reference value we apply readout error mitigation (REM) to single qubit measurements at the end of the circuit (green), the result of which (red) is in good agreement with our corrected values.}
	\label{fig:REM}
\end{figure}

The results of simulating this procedure are shown in Fig.~\ref{fig:REM}. As can be seen, the results differ for the virtual ibm\_cairo (Fig.~\ref{fig:REM_cairo}) and ibm\_hanoi-devices (Fig.~\ref{fig:REM_hanoi}), due to differences in their noise model.
However, in both cases, the values reconstructed using our correction factor and REM are in good agreement, lending credence to the proposition that our method should also perform well in real experiments.

\subsection{Comparison to Experimental Data}

Finally, we have implemented our method on the ibm\_hanoi-backend provided via the IBM Quantum Experience. Fig.~\ref{fig:real} shows the results obtained. These data can be compared to those obtained via simulation in Fig.~\ref{fig:REM_hanoi}. 

\begin{figure}[tp]
	\centering
	\subfloat[Initial state $\ket{00\ldots 0}$ \label{fig:real}]{\begin{overpic}[width=0.5\linewidth]{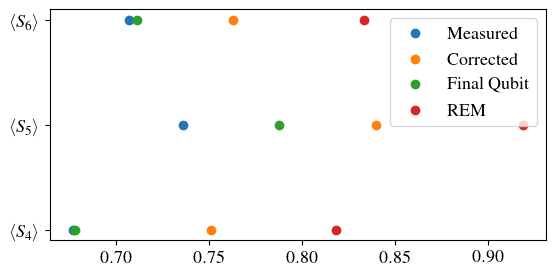}\end{overpic}}\\
	\subfloat[Initial state $\ket{11\ldots 1}$	\label{fig:all1}]{\begin{overpic}[width=0.5\linewidth]{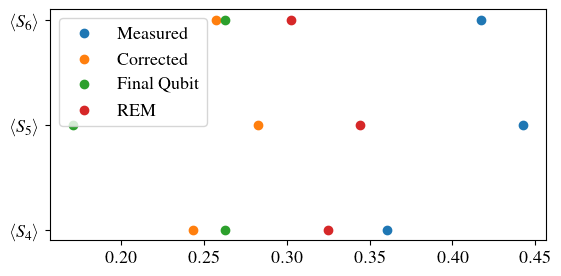}
			
		\end{overpic}}\\
	\subfloat[Average of the two results from (a) and (b), which removes the asymmetry of the noise.\label{fig:mix}]{\begin{overpic}[width=0.5\linewidth]{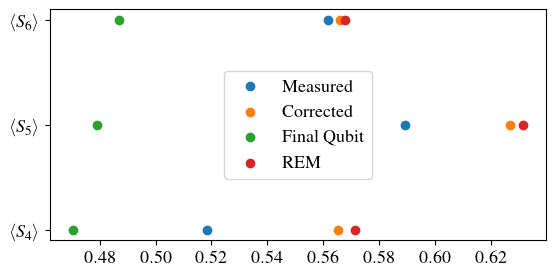}
	\end{overpic}}
	\caption{Measurement of the three stabilizer generators $S_4$, $S_5$, and $S_6$, as carried out on the ibm\_hanoi-device.}
	
\end{figure}

In comparison with the simulated results, the agreement between the values obtained via REM and via our method is substantially worse.
However, it is not necessarily clear which of these values should be regarded as `more accurate': both represent only a reconstruction of the value after it has been disturbed by the noisy measurement.

In the case of the ibm\_hanoi-device, it is interesting to note that the value obtained via the final single qubit measurements is comparable to that from the uncorrected stabilizer measurement, or in one case even closer to the ideal value. 
As the former is subject to additional noise from the stabilizer measurement, and is itself inherently noisy, one might instead have expected a value farther from the ideal.
A possible reason for this is that the ideal value is obtained in the state $\ket{00\ldots 0}$, and the measurement errors for the IBM devices, due to relaxation effects, are generally asymmetric in nature, with the likelihood of a mis-measured $0$ exceeding that of mis-measuring $1$. 

This interpretation is bolstered by the fact that, as compared to the ibm\_cairo-device, ibm\_hanoi suffers less from propagated errors, with the noise introduced in measurement thus dominating.
This can be inferred from the fact that the corrected values are larger than the measured ones. 
This implies $\alpha > 1$, and thus, $\gamma < \beta$, meaning that the measurement more strongly impacts the observed expectation value.

To investigate the influence of the asymmetric noise, we repeated the above measurement procedure, starting instead with the state $\ket{11\ldots 1}$. 
Fig.~\ref{fig:all1vall0} shows a comparison of the quasi-probabilities for the 15 states providing the highest contribution to the result when initializing the ibm\_hanoi-device in the state $\ket{11\ldots 1}$ (blue) versus $\ket{00\ldots 0}$ (orange). 
As can be seen, the state $\ket{11\ldots 1}$ shows a significantly greater amount of decay, with significant probability mass accumulating on the `rest' of the states, i.e. those beyond the 15 highest contributors.

\begin{figure}[tp]
	\centering
	\begin{overpic}[width=0.5\linewidth]{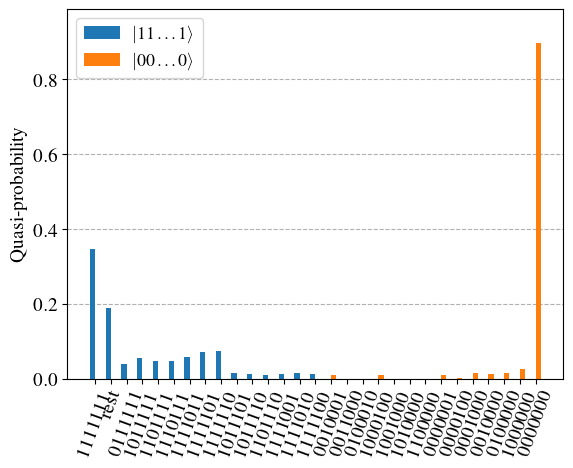}
		
	\end{overpic}
	\caption{Comparison for an example measurement run of the quasi-probability distribution of results when initializing into the state $\ket{00\ldots 0}$ versus $\ket{11\ldots 1}$ on the ibm\_hanoi-device. The states are ordered by their hamming-distance to the target state.}
	\label{fig:all1vall0}
\end{figure}

Accordingly, as Fig.~\ref{fig:all1} shows, for the ibm\_hanoi-device again initialized to the $\ket{11\ldots 1}$-state, we observe a significantly noisier outcome for the considered stabilizer generators. 
Moreover, in contrast to the $\ket{00\ldots 0}$-case, the measured values now tend to overestimate the true values.
As the figure shows, this behavior is correctly captured, and partially corrected for, by our procedure.

To symmetrize the effect of the noise we apply the logical operation $X^{\otimes n}$ on the initial state in half of the runs. 
This interchanges the roles of $0$ and $1$, but ideally this does not affect the measurement result.
However, it removes the bias of the relaxation towards the logical $\ket{0}$-state.
This ensures that the assumption of Pauli noise, which is symmetric w.r.t. to $0$ and $1$, is better justified.
Equivalently, we consider the incoherent mixture $\frac{1}{2}\left(\ket{00\ldots 0}\!\!\bra{00\ldots 0} + \ket{11\ldots 1}\!\!\bra{11\ldots 1}\right)$ by simply averaging over the obtained results. 
As can be seen in Fig.~\ref{fig:mix}, this yields a significantly better agreement between the results obtained via our method and the reconstructed `true' values obtained via readout-error mitigated final qubit measurements.
Note, however, that our results were obtained without enforcing Pauli noise by randomized compiling.
\section{Conclusion}
\label{sec:conclusion}

We have introduced and elaborated on a method that allows using a single calibration experiment to efficiently find the noise introduced by stabilizer measurements during error syndrome extraction. 
During such a measurement, two sources of errors are relevant: the readout error affecting the observed value, and the propagated error adding to the noise introduced into the quantum circuit.
This means that, e.g., the expectation value of a stabilizer generator calculated from the measurement results generally differs from the `true' value due to the noise added during the measurement. 
However, it is this value that is most useful for applications like estimating a circuit's total noise channel, or informing subsequent error correction operations.

The proposed method is to simply carry out the stabilizer measurement to be calibrated twice in succession with a known input state.
The information thus obtained then allows computing an estimate of the expectation value of the stabilizer after measuring it has introduced additional noise into the quantum circuit. 
The estimate can be conditioned on the observed syndrome or averaged over all outcomes.

Making use of the conditioned expectation values has the potential of improving several common tasks related to the performance of quantum computing workloads. 
For one, we could show that the accuracy of the estimate of the total noise affecting a quantum circuit, performed using the method introduced by Wagner et al.~\cite{wagner2022pauli}, can be greatly increased.
Additionally, in cases where the standard approach leads to an improper correction operation for errors introduced during execution, using the corrected expectation value yields better outcomes, as we have discussed using the example of the $[[7,1,3]]$-Steane code.

For our method to be applicable, we have to assume that the noise introduced by a stabilizer measurement is independent of its position in the circuit. 
To make this assumption plausible on currently available NISQ devices, care has to be taken: as these devices are typically of limited qubit connectivity, during execution, swap-operations have to be used to facilitate the execution of two-qubit gates on arbitrary circuit-qubits. 
However, these operations generally change the mapping between circuit-qubits and the physical qubits of the backend; but then, any repeated measurement will be implemented in a different way physically, calling the assumption into question. 

To remedy this, we have proposed an operations schedule that implements the stabilizer measurements in such a way as to leave the mapping between circuit- and backend-qubits invariant. 
Using this, we have shown in numerical simulations that our method recovers the `true' expectation value under various different noise models.
Moreover, running calibration and measurement circuits on a real quantum backend accessed via the IBM Quantum Experience, we could demonstrate the production of better guesses for the expectation value in a real-world setting.

There are several possible ways to further develop the technique, and apply it to different tasks. 
For further experimental work, there are error correcting codes optimized for the `Heavy Hexagon'-layout of IBM's quantum devices~\cite{chamberland2020topological}, which might help ameliorate the problem of having to introduce an excess of swap-operations.
Additionally, the general strategy of using calibration measurements to find the effect of measurement-induced noise may be applied beyond the setting of error correction. 
A possibility here seems to be the use of entanglement witnesses: While their measurement is used to characterize the entanglement present in a quantum state, the noise introduced during the measurement may itself reduce it. 
Hence, estimating the `true' value after this noisy measurement may lead to better lower bounds on the entanglement still present, and thus, usable for e.g. quantum communication tasks.

\begin{acknowledgments}
This work was supported by the QuantERA grant EQUIP via DFG project 491784278, by the Federal Ministry for Economics and Climate Action (BMWK) via project R-QIP and by the DLR project ELEVATE. The funders played no role in study design, data collection, analysis and interpretation of data, or the writing of this manuscript. The initial idea for this project was developed in discussions between the authors and Thomas Wagner, Hermann Kampermann and Dagmar Bruß. We thank Thorge Müller for his comments and support. We acknowledge the use of IBM Quantum services for this work. The views expressed are those of the authors, and do not reflect the official policy or position of IBM or the IBM Quantum team.
\end{acknowledgments}

\section*{Author contributions}
ME organized the project and obtained the initial theoretical results, which CW and ME subsequently generalized. JS designed and carried out the experiments. All authors contributed to the interpretation of the data and wrote the manuscript together.

\section*{Data availability statement}
The data containing the numerical results presented in Section~\ref{sec:experiment} are archived by the authors institution and will be shared on request to the corresponding author. 

\section*{Competing Interests}
The authors declare no competing interests.

\bibliography{noiseestimation} 
\onecolumngrid
\clearpage
\section*{Supplemental Material}
\appendix

\section{Mathematical background for our calibration procedure}
\label{app:details}

The goal of this appendix is to give a rigorous treatment of syndrome measurements affected by Pauli noise at a conceptually appropriate level of generality. We begin by setting up notation and recalling some relevant mathematical notions.
While the presentation is certainly influenced by elementary representation theory, no knowledge of it is required to follow the statements.

\subsection*{Preliminaries}
	Throughout this section the following data is fixed:
	A finite dimensional complex inner product space $V$ together with a finite \emph{error group} $E \leq U(V)$ of unitary operators and a finite abelian \emph{stabilizer group} $S$ acting on $V$ via operators in $E$. 
	The associated \emph{index group} is the quotient 
	\(
		G\coloneqq E/(E\cap U(1))
	\)
	by the scalars (embedded as multiples of $\Id_V$). The index group $G$ will be assumed to be abelian. In that case there is a well-defined 'commutator' pairing
	\(
		\omega:G\times G\rightarrow U(1)
	\)
	determined by the relation $e_1e_2=\omega(e_1,e_2)e_2e_1$ in $E$.

\begin{notation}
	The action referred to above is encoded by a group homomorphism $\rho_S\colon S\rightarrow E$. It will be left implicit and we simply write $sv=\rho_S(v)$.
\end{notation}

\begin{rem}\label{rem:errorgroups}
	For the purposes of this appendix, the above terms error- and stabilizer group can just be taken as names.
	We recall their usual meaning  (e.g.\ see \cite{knill1996group}, \cite{klappenecker2002beyond}): the trace $\tr e=0$ vanishes for all non-scalar $e \in E$ ($[e]\neq 1\in G$) and $(\dim V)^2=\vert G\vert$ holds. Equivalently, if $\{\rho_V(g)\}_{g\in G}\subset E$ is a complete set of representatives, then it (and hence any other) forms a \emph{unitary error basis} on $V$, that is an orthogonal basis of the linear operators $\cB(V)$ with respect to the trace inner product.	
	
	The action of the stabilizer group should contain no scalars other than the identity: $sv= \lambda v$ for some $\lambda \in U(1)$ and all $v\in V$ implies $\lambda=1$ (also see Remark~\ref{rem:dimensions} below).
\end{rem}
\begin{example}
	The case of interest in the main body of text is $V=(\bC^2)^{\otimes n}$ with $E=P_n$ the Pauli group
	and $G=P_n/\{\pm 1, \pm i\}$ the effective Pauli group.
\end{example}

\begin{rem}
	Although only the above explicit case is used in this paper,
	we prefer to work at this level of generality and in a coordinate free way since it conceptually clarifies the situation  without imposing much additional work. This also avoids some of the baggage entailed by the necessary bookkeeping.
	We also note that all combinations of modular qudits as basic building blocks are covered here.
\end{rem}

\begin{notation}\begin{itemize}
	\item 
	We will use $(\blank)$ as placeholder notation, i.e.\ for a function $f$ of two variables $f(a,\blank)$ is the map $b\mapsto f(a,b)$.
	\item
	If $A$ is a group, we write \(\widehat A\coloneqq \operatorname{Hom}(A, U(1))\) for the dual group consisting of the \emph{characters} on $A$, i.e.\ the homomorphisms to the circle group \(U(1)\leq \mathbb C^\times\) with point-wise multiplication. The unit element (the constant function with value 1) is also denoted by 1.
	\item
	Multiplicative notation will be used for all groups with the exception of $\ftwo^m$.
	\item
	Let $P$ be a probability distribution on a finite set $X$ with mass function $P\colon X\rightarrow [0,1]$.
	The expectation value of a random variable $f\colon X\rightarrow \bC$ (i.e.\ a function) with respect to $P$ will be denoted by
	\(
	E(f;P) \coloneqq \sum_{x\in X}f(x)P(x).
	\)
	\end{itemize}
\end{notation}

\begin{prop}\label{prop:isotypical_decomposition}
As a $S$-representation
\(
	V=\bigoplus_{\chi \in \widehat S}V_{\chi}
\)
decomposes into the orthogonal sum of its isotypical components
\[
	V_{\chi}\coloneqq\{v\in V\vert sv=\chi(s)v\ \forall s\in S\}\subseteq V
\]
with associated orthogonal projections
\[
	\pi_{\chi}\coloneqq \frac{1}{\vert S\vert} \sum_{s\in S}\overline{\chi(s)}s.
\]
Put differently, this means that $V_{\chi}$ is the common $\chi$-eigenspace of $S$ and the $\pi_{\chi}$ form a complete set of operators for a projective measurement:
\[
	\sum_{\chi\in \widehat S} \pi_{\chi}=\Id_V.
\]
\end{prop}
\begin{proof}
This follows by direct computation from the orthogonality relations
	\(
		\sum_{s\in S}\chi(s)=0	
	\)
	and
	\(
		\sum_{\phi\in \widehat S}\phi(t)=0	
	\)
	for any non-trivial character $\chi\neq 1$ and element $t\neq 1\in S$,
	e.g.\ see \cite{Konrad} or any textbook on representation theory.
\end{proof}

Given a state $\rho\in \mathcal B(V)$ (i.e.\ $\rho\geq 0$ is positive and $\tr \rho=1$), we write
\[
	p(\chi\vert \rho)\coloneqq \tr(\pi_{\chi}\rho)
\]
for the probability of obtaining the result $\chi\in \widehat S$ in an ideal measurement.
Moreover, the expectation value
\[
	\expvalideal{s}{\rho} \coloneqq E(\ev_s;p(\blank\vert \rho)) = \sum_{\chi}\chi(s)p(\chi\vert \rho) = \tr (s\rho)
\]
of the random variable $\ev_s\colon \widehat S\rightarrow \bC$, $\ev_s(\chi)=\chi(s)$ agrees with the trace on the right-hand side since the projections sum to the identity.
\begin{rem}\label{rem:dimensions}
	\begin{itemize}
		\item
		In the situation of Remark~\ref{rem:errorgroups} the non-scalar condition on the action of $S$ ensures that $\dim V_\chi= \frac{\dim V}{\vert S\vert}\vert S^{\operatorname{triv}}\vert$ holds for the eigenspaces described above, where $S^{\operatorname{triv}}\leq S$ is the subgroup of elements acting as the identity. This observation amounts to computing the trace in the projections formula.
		\item			
		While $s$ will not be an observable in general (unless $s^2=\Id_V$, e.g.\ in the case of qubits), this is formally not an issue. We also note that real and imaginary parts can be written as expectation values of the commuting observables $\frac{1}{2}(s+s^\dagger)$ and $\frac{i}{2}(s^\dagger-s)$.
	\end{itemize}
\end{rem}

\noindent
Let $A$ be a finite abelian group. The \emph{Fourier transformation}
\[
	\cF:\operatorname{Fun}(\widehat A,\mathbb C)\overset{\cong}{\longrightarrow} \operatorname{Fun}(A,\mathbb C)
\]
of complex valued functions on $\widehat A$ (and its inverse) will take the form 
\[
	\quad\cF[f](a)= \sum_{\chi \in \widehat A}\chi(a)f(\chi),\quad\cF^{-1}[g](\chi)= \frac{1}{\vert A\vert}\sum_{a\in A}\overline{\chi(a)}g(a).
\]
The \emph{convolution} \(	f\ast g\colon \widehat A\rightarrow \mathbb C	\) of two functions $f,g$ on $\widehat A$ is given by 
\[
	(f\ast g)(\chi)=\sum_{\phi\in\widehat A}f(\phi)g(\phi^{-1}\chi)
\]
and we have the important relation
\[
	\cF[f\ast g]= \cF[f]\cdot\cF[g]
\]
between convolution and pointwise multiplication.

We record the following standard but useful observation, also used in \cite{wagner2022pauli} for example, which just amounts to spelling out the definitions:
\begin{lemma}\label{lem:expval_fourier}
For a probability distribution $P$ on $\widehat A$
and $a\in A$ the expectation value
\[
	E(a;P)\coloneqq E(\operatorname{ev}_a;P) = \sum_{\chi\in \widehat A}\chi(a)P(\chi)=\cF[P](a)
\]
of the random variable $\operatorname{ev}_a$ on $\widehat A$ agrees with the Fourier transformation of $P$ evaluated at $a$.
In particular, given the state $\rho$, we have the equality
\[
	\expvalideal{s}{\rho}=\cF[p(\blank\vert\rho)](s)
\]
for the expectation value of $s\in S$.
\end{lemma}

\begin{notation}\label{not:conv}
	\begin{itemize}
	\item 
	When several variables are around, the convolution with respect to one of these will be indicated by a subscript:
	\(
	(f(a,b)\ast_a g(a,b))(x)=(f(\blank,b) \ast g(\blank,b))(x).
	\)
	\item
	Given a linear operator $A\colon V\rightarrow W$ to another Hilbert space $W$ and $\rho\in \cB(V)$ we write
	\[
		A_\ast(\rho)\coloneqq A\rho A^\dagger\in \cB(W)
	\]
	and note that this only depends on $A$ up to multiplication with unit scalars. In particular, the projective unitary group $PU(V)=U(V)/U(1)$ and hence the index group $G$ acts
	on $\cB(V)$ via
	\[
	g.\rho \coloneqq e\rho e^{-1}=e_{\ast}(\rho),
	\]
	where $e\in E$ is any element with $g=[e]$. By slight abuse we will mix these notations and write terms like $(gA)_\ast(\rho)$.
	\item
	For a finite set $X$ the $\bC$-linearization of $X$ is denoted by $\bC\{X\}$.
	It is the complex inner product space with orthonormal basis $\{\ket x\}_{x\in X}$ given by $X$.
	The corresponding dual basis consisting of linear forms is $\{\bra x\}_{x\in X}$.
	\end{itemize}
\end{notation}

\subsection*{Faulty measurements}

In order to reason about faulty measurements as single high level gadgets in the quantum circuit model,
we will define them operationally as \emph{quantum instruments}, i.e.\
quantum channels (completely positive and trace preserving linear maps)
\(
	\mathcal I\colon\cB(V)\rightarrow\cB(V\otimes \bC\{X\})
\)
of the form
\[
	\mathcal I(\rho)= \sum_{x\in X} \cE_x(\rho)\otimes\ket{x}\bra{x}.
\]
 Here the \emph{component maps} $\cE_x\colon \cB(V)\rightarrow \cB(V)$ are implicitly determined as the composition
\[
	\cE_x = (\Id_V\otimes \bra x)_{\ast}\circ \mathcal I.
\]
We recall that $\mathcal I$ describes a generalized measurement as follows:
given a state $\rho$, the probability of obtaining the result $x\in X$ and the corresponding post-measurement state are
\[
	p_{\mathcal I}(x\vert \rho) \coloneqq \tr{\cE_x(\rho)}\quad\text{respectively}\quad\rho^{\mathcal I}_x\coloneqq\frac{1}{p_{\mathcal I}(x\vert \rho)}\cE_x(\rho),
\]
where consideration of $\rho^{\mathcal I}_x$ implicitly assumes non-vanishing probability. The instrument will usually be omitted from the notation and we simply write $\tilde p(x\vert \rho)$ and $\rho_x$.

Let \(\Phi\in U(V\otimes \linstab)\) be the unitary operator uniquely specified by
\[
	 \Phi(v\otimes \ket{\phi})\coloneqq \sum_{\chi\in\widehat S}\pi_{\chi}(v)\otimes \ket{\phi^{-1}\chi}
\]
for $v\in V$ and $\phi \in \widehat S$. This is the abstract result of performing an ancilla circuit (at least with the ancilla-factor initialized to \ket 1).

\label{def:faulty_measurement}
\begin{definition}
	The \emph{faulty measurement} $\widetilde M_{\cE}$ with internal error channel
	\[
		\cE\colon\cB(V\otimes\linstab)\rightarrow \cB(V\otimes \linstab)
	\]
	is defined as the quantum instrument given by the composition
	\[
		\cB(V)\overset{\Phi_{\anc}}{\longrightarrow} \cB(V\otimes \linstab)\overset{\cE}{\longrightarrow}\cB(V\otimes \linstab)\overset{M}{\longrightarrow}\cB(V\otimes \linstab)
	\] of CPTP maps, where
	\(
		\Phi_{\anc}=\Phi(-\otimes \ket 1)_\ast
	\)
	and
	\(
		M= \sum_{\chi}(\Id_V\otimes \ket{\chi}\bra{\chi})_{\ast}
	\)
	is the measurement in the $\widehat S$-basis. For $s \in S$ we denote by
	\[
	\expvalnoisy{s}{\rho}\coloneqq E(\ev_s;\tilde p(\blank\vert\rho))=\cF[p(\blank\vert\rho)](s)
	\]	
	the expectation value with respect to the faulty measurement (cf.\ Lemma~\ref{lem:expval_fourier}).
\end{definition}

From now on we restrict to \emph{diagonal error channels}, the generalization of Pauli channels.
The shift and multiply operators
\(
	X_{\phi}, Z_s\in U(\bC\{\widehat S\})
\)
defined for $\phi\in\widehat S$, $s\in S$ by
\[
	X_{\phi}(\ket\chi)\coloneqq\ket{\phi\chi},\quad Z_s(\ket\chi)\coloneqq\chi(s)\ket{\chi}
\]
generate an error group $E_S\leq U(\bC\{\widehat S\})$ with $E_S/(E_S\cap U(1))\cong\widehat S\times S$ as the associated index group.
A diagonal channel $\cE$ on $V\otimes \bC\{\widehat S\}$ is a channel of the form
\[
	\cE = \sum_{(g,\phi,s)\in G\times \widehat S\times S}P(g,\phi,s) (g\otimes (\phi, s))_\ast,
\]
for a probability distribution $P$ on
\(
	G\times \widehat S\times S.
\)

\begin{rem}
	The $X$ and $Z$ operators form a unitary error basis on $\bC\{\widehat S\}$.
	If $E\leq U(V)$ is an actual error group in the sense of Remark~\ref{rem:errorgroups}, then so is
	$E\otimes E_S\leq U(V\otimes \bC\{\widehat S\})$ with associated  index group $G\times \widehat S\times S$.
	In that case the channel considered as an endomorphism of $\cB(V\otimes \bC\{\widehat S\})$ is diagonal with respect to a unitary error basis indexed by $G\times \widehat S\times S$.
\end{rem}
The $Z$-errors are automatically averaged out in the faulty measurement, as one sees from the formula below.

\begin{prop}\label{prop:faulty_measurement}
	For a diagonal error channel the equality
	\[
		\widetilde M_\cE(\rho)=\cE(\Phi_{\anc}(\rho))= \sum_{g\in G, \phi\in \widehat S}P(g,\phi)(g\pi_{\phi^{-1}\chi})_\ast(\rho)\otimes \ket{\chi}\bra{\chi}
	\]
	holds, where $P(g,\phi)= \sum_{s\in S}P(g,\phi,s)$.
\end{prop}
\begin{proof}
	We simply write out the definitions and note that in the second to last equality the $\chi(s)$-factors introduced by $Z_s$ have canceled out:
	
	\begin{align*}
		\widetilde M_\cE(\rho)&=M\left(\cE\left(\sum_{\chi}\pi_{\chi}(v)_\ast(\rho)\otimes \ket{\chi}\bra{\chi}\right)\right)\\
		&=M\left(\sum_{g, s, \phi, \chi}P(g,\phi,s)(g\pi_{\chi}(v))_\ast(\rho)\otimes (\phi, s)_\ast(\ket{\chi}\bra{\chi})\right)\\
		&=M\left(\sum_{g, s, \phi, \chi}P(g,\phi,s)(g\pi_{\chi}(v))_\ast(\rho)\otimes X_\phi Z_s\ket{\chi}\bra{\chi}Z_{s^{-1}}X_{\phi^{-1}}\right)\\
		&=M\left(\sum_{g, s, \phi, \chi}P(g,\phi,s)(g\pi_{\chi}(v))_\ast(\rho)\otimes \ket{\phi\chi}\bra{\phi\chi}\right)\\
		&= \sum_{g,\phi,\chi}P(g,\phi)(g\pi_{\phi^{-1}\chi})_\ast(\rho)\otimes \ket{\chi}\bra{\chi}.
	\end{align*}
\end{proof}

So there is no loss of generality in the following

\begin{definition}\label{def:faulty_measurement_pauli}
	The faulty measurement $\widetilde M_P$ with internal probability distribution $P$ on $G\times \widehat S$
	is the faulty measurement with internal error channel 
	\[
		\cE_P = \sum_{g\in G,\phi\in\widehat S}P(g,\phi) (g, \phi)_\ast.
	\]
\end{definition}

\subsection*{Behavior of faulty measurements}
With all the necessary setup in place we now turn towards a more quantitative analysis of the behavior of faulty measurements.
\begin{prop}\label{prop:probability_expectedvalues}
	The probability of obtaining $\chi\in \widehat S$ is given by the convolution
	\[
		\tilde p(\chi\vert \rho)=(P\ast p(\blank\vert \rho))(\chi),
	\]
	where $P(\phi)=\sum_{\in G} P(g,\phi)$. In this case $\rho$ is mapped to the new state
	\[
		\rho_\chi = \frac{1}{\tilde p(\chi\vert\rho)}\sum_{g\in G,\phi\in \widehat S}P(g,\phi)(g\pi_{\phi^{-1}\chi})_\ast(\rho).
	\]
	Moreover, for $s\in S$ the expectation values with respect to ideal and faulty measurement are related via the formula
	\[
			\expvalnoisy{s}{\rho} = E(\ev_s;P)\expvalideal{s}{\rho}.
	\]
\end{prop}
\begin{proof}
	This amounts to applying the trace in the formula of Proposition~\ref{prop:faulty_measurement}.
	The statement about the expectation values (see Lemma~\ref{lem:expval_fourier}) is a corollary,
	since the Fourier transformation turns the convolution into a pointwise product.
\end{proof}

\begin{cor}\label{cor:inverse_fourier}
The ideal probability distribution can be reconstructed via the inverse Fourier transformation:
\begin{align*}
	p(\chi\vert\rho) &= \cF^{-1}[s\mapsto E(\ev_s;P)^{-1}\expvalnoisy{s}{\rho}](\chi) \\
	&= \frac{1}{\vert S\vert}\sum_{s\in S} \overline{\chi(s)}E(\ev_s;P)^{-1}\cdot \expvalnoisy{s}{\rho}.
\end{align*}

\end{cor}

Performing two consecutive faulty measurements with respect to commuting stabilizer groups can formally be regarded as a single measurement in the following sense. By the composition of two quantum instruments
\(
	\mathcal I\colon\cB(V)\rightarrow\cB(V\otimes \bC\{X\})
\)
and
\(
\mathcal J\colon\cB(V)\rightarrow\cB(V\otimes \bC\{Y\})
\)
we mean the quantum instrument
\[
	\mathcal J \diamond \mathcal I\colon\cB(V)\rightarrow\cB(V\otimes \bC\{Y\times X\})
\]
with component maps
\[
	\cE^{\mathcal J \diamond \mathcal I}_{y,x}(\rho)=\cE^{\mathcal J}_y(\cE^{\mathcal I}_x(\rho))
\]
for $(y,x)\in Y\times X$.
Up to identification this is just the composition of the CPTP-maps $\mathcal J\otimes \Id_{\bC\{X\}}$ and $\mathcal I$.

Let $T$ be another stabilizer group whose action commutes with that of $S$ (i.e.\ $stv=tsv$ for all $s\in S$, $t\in T$, $v\in V$) and let $Q$ be a probability distribution on $G\times T$. We recall that since the index group $G$ is abelian, there is a commutator pairing
\(
	\omega:G\times G\rightarrow U(1).
\)
In particular, for a fixed $g\in G$ it restricts to a character
\[
\omega(g,\blank)\colon S\rightarrow U(1)
\]
on $S$ and for $\chi\in \widehat S$ the identity
\[
g\pi_\chi=\pi_{\omega(\blank,g)\chi} g
\]
holds. In the following statement we indicate the stabilizer groups with additional suggestive subscripts.
\begin{prop}\label{prop:composition_measurements}
The composition
\(
	\widetilde M_{T;Q}\diamond\widetilde M_{S;P}
\)
of quantum instruments is the faulty measurement
\[
	\widetilde M_{T\times S;Q\ast^{\omega}P} \colon \cB(V)\longrightarrow\cB(V\otimes\bC\{\widehat T\times \widehat S\})
\]
with internal probability distribution $Q\ast^{\omega} P$ defined on 
\(
G\times \widehat T\times \widehat S
\)
as the '$\omega$-twisted' convolution
\[
\quad (Q\ast^{\omega} P)(g,(\phi_2,\phi_1))= \sum_{k\in G} Q(gk^{-1},\omega(k,\blank)\phi_2)P(k,\phi_1).
\]
\end{prop}
\begin{proof}
Using the identity $g\pi_\chi =	\pi_{\omega(\blank,g)\chi} g$, it follows by direct calculation from the formula in Proposition \ref{prop:faulty_measurement} that the component map indexed by $(\chi_2,\chi_1)\in \widehat T\times \widehat S$ is given by
\[
	\sum_{g_1,g_2\in G, \phi_1\in \widehat S, \phi_2\in \widehat T} Q(g_2, \phi_2\omega(g_1,-))P(g_1,\phi_1)(g_2g_1\pi^T_{\phi_2^{-1}\chi_2}\pi^S_{\phi_1^{-1}\chi_1})_{\ast}.
\]
We note that $\pi^T_\chi\pi^S_{\eta}=\pi^{T\times S}_{(\chi, \eta)}$ is the orthogonal projection with respect to the action of $T\times S$, as can be seen from the projection formula in Proposition~\ref{prop:isotypical_decomposition}. After reindexing the sum we end up with the component map of the faulty measurement described in the statement above.
\end{proof}

\begin{lemma}
	For $s\in S$ the Fourier transformation of
	\(
		\phi \mapsto \phi(s)p(\phi\vert \rho)
	\)
	is given by
	\[
		\cF[\ev_s\cdot p(\blank\vert\rho)](t)=\expvalideal{st}{\rho}.
	\]
\end{lemma}
\begin{proof}
	This is an instance of the general formula $\cF[\ev_s\cdot f](t)=\cF[f](st)$.
\end{proof}

We now establish the main result of this section about the relationship between expectation values before and after measurement.
It is expressed via further expectation values with respect to the conditional probability distributions
$P_{\phi}$ on $G$ defined by
\[
	P_\phi(g) =\frac{P(g,\phi)}{P(\phi)},
\]
namely
\[
	E(s;P_\phi)\coloneqq E(\omega(s,\blank); P_\phi)=\frac{1}{P(\phi)}\sum_{g\in G} \omega(s,g)P(g,\phi).
\]

\begin{thm}\label{thm:expected_values}
	For $s\in S$ and $\chi \in \widehat S$ the ideal expectation value of the conditional state $\rho_\chi$ can be written as the convolution
	\begin{align*}
		\expvalideal{s}{\rho_\chi}\tilde p(\chi\vert\rho) &= \left[E(s;P_\phi)P(\phi))\ast_{\phi}(\phi(s)p(\phi\vert\rho)\right](\chi) \\
		&= \left(E(s;P_\phi)P(\phi))\ast_{\phi}\cF^{-1}[\expvalideal{s\cdot(\blank)}{\rho}](\phi)\right)(\chi).
	\end{align*}
	In particular, for a state $\rho$ with nowhere vanishing $s\mapsto \expvalideal{s}{\rho}$ the formula

	\begin{align*}
		\expvalideal{s}{\rho_\chi} &= \cF^{-1}[t\mapsto \frac{\cF[\psi \mapsto \expvalideal{s}{\rho_\psi}\tilde p(\psi\vert\rho)](t)}{\expvalideal{st}{\rho}}](\chi) \\
		&= \frac{1}{\vert S\vert}\sum_{t\in S, \phi \in \widehat S}\bar \chi(t)\phi(t)\frac{\expvalideal{s}{\rho_\phi}\tilde p(\phi\vert\rho)}{\expvalideal{st}{\rho}}
	\end{align*}
	holds.
\end{thm}
\begin{proof}
	Using $se\pi_\eta=\omega(s,g)\eta(s)e\pi_\eta$ for $g=[e]\in G, \eta \in \widehat S$ we calculate
	\begin{align*}
		E(s\vert \rho_\chi)\tilde p(\chi\vert\rho) &= \tr\left( \sum_{g,\phi}P(g,\phi)s(g\pi_{\phi^{-1}\chi})_\ast(\rho) \right) \\
		&= \tr\left( \sum_{g,\phi}P(g,\phi)\omega(s,g)(\phi^{-1}\chi)(s)(\pi_{\phi^{-1}\chi})_\ast(\rho) \right) \\
		&= \sum_{g,\phi}P(g,\phi)\omega(s,g)(\phi^{-1}\chi)(s)p(\phi^{-1}\chi\vert\rho) \\
		&= \left( \sum_{g}P(g,\phi)\omega(s,g)\ast_\phi((\phi)(s)p(\phi\vert\rho)) \right)(\chi) \\
	\end{align*}
	to obtain the convolution formula. The inverse Fourier transformation then appears due to the previous lemma.
\end{proof}
\begin{rem}
	This can also be deduced from Proposition~\ref{prop:composition_measurements}, where the second measurement is taken to be the ideal stabilizer measurement with respect to $T=S$.
\end{rem}
\begin{rem}	\label{rem:qubits}
	We make the above explicit in the case $V=(\bC^2)^{\otimes n}$ and $E=P_n$ the Pauli group with stabilizer subgroup $S=\langle S_1,\ldots, S_m\rangle$ generated by fixed generators determining an isomorphism
	\[
		\ftwo^m\overset{\cong}{\longrightarrow} S, \quad a=(a_1,\ldots,a_m) \mapsto S(a)\coloneqq S_1^{a_1}\cdots S_m^{a_m}.
	\]
	The character group $\ftwo^m\cong\widehat{\ftwo^m}$ is identified via $a\mapsto \chi_a$, where $\chi_a(b)=(-1)^{a\cdot b}$.	
	In this sense the evaluation pairing $\widehat S\times S\rightarrow U(1)$ corresponds to the inner product on $\ftwo^m$.
	Fourier transformation and convolution simplify to
	\[
		\cF[f](a)=\sum_{x\in \ftwo^m}(-1)^{a\cdot x}f(x)\quad\text{respectively}\quad (f\ast g)(x)=\sum_{u\in \ftwo^m}	f(u+x)g(u).
	\]
	We have so far carefully distinguished the conceptual roles of index and error group. To connect to the main body of text, we now implicitly choose a unitary error basis $E'\subset E$ indexed by $G$ (e.g.\ the standard $X$- and $Z$-operators) and instead of $g\in G$ let the variable $e$ run over $E'$ in the following formulas.
	The expressions of Proposition \ref{prop:probability_expectedvalues} for probabilities and post-measurement state then take the form (with $x\in\ftwo^m$)
	\[
		\tilde p(x\vert \rho)=(P\ast p(\blank\vert \rho))(x)= \sum_{u \in \ftwo^m}P(u)p(x+u)
	\]	
	and
	\[
		\rho_x = \frac{1}{\tilde p(x\vert\rho)}\sum_{e\in E',u\in \ftwo^m}P(e,u)(e\pi_{x+u})_\ast(\rho).
	\]
	The inversion formula of Theorem \ref{thm:expected_values} for $a,u\in \ftwo^m$ reads as:
	\[
		E(a; P_u)P(u)=\frac{1}{2^m}\sum_{b, x\in\ftwo^m}(-1)^{(u+x)b}\frac{\expvalideal{S(a)}{\rho_x}\tilde p(x\vert\rho)}{\expvalideal{S(a+b)}{\rho}}.
	\]
\end{rem}

To conclude, we explain how some interpolation between the extreme cases of recording every single syndrome or averaging over all can be achieved. For a subset $X\subseteq \ftwo^m$ consider the state
\[
	\rho_X\coloneqq \sum_{x\in X}\rho_x\tilde p(x\vert \rho)=\frac{1}{\tilde p(X\vert \rho)}\sum_{x\in X, e, u}P(e,u)(e\pi_{x+u})_\ast\rho
\]
obtained in the situation that a result in $X$ was measured but the specific one not recorded. If $X$ is suitably choosen, the inversion formula can be generalized.
\begin{cor}
	Let $W\subseteq \ftwo^m$ be a $k$-dimensional $\bF_2$-linear subspace. Writing
	\(
		[u] = u+W
	\)
	for the cosets we have the equality
	\[
		E(a; P_{[u]})P([u])=2^{k-m}\sum_{b\in W^\perp,[x]\in \bF_2^m/W}(-1)^{(u+x)b}\frac{E(a\vert \rho_{[x]})\tilde p([x]\vert\rho)}{E(a+b\vert \rho)}.
	\]
	In particular, for $W=\ftwo^m$ this simplifies to
	\[
	E(a; P_{\ftwo^m})=\frac{E(a\vert \rho_{\ftwo^m})}{E(a\vert \rho)}.
	\]
\end{cor}
\begin{proof}
	This amounts to summing up both sides of the formula over a coset. 
\end{proof}
\begin{rem}\label{ex:maxlike}
	We spell out some details regarding the maximum-likelihood decoding in Section~\ref{sec:depolarizing_noise}.
	Suppose that the Pauli channel $\cE$ is described by the probabilities $Q(e)$. 
	Via syndrome measurements we have access to the expectation values (for a stabilizer state $\rho$)
	\[
		\expvalideal{s}{\cE(\rho)}=\sum_e Q(e)(-1)^{(s,e)}=\sum_y\left(\sum_{\Syn e = y}Q(e)\right)(-1)^{a\cdot y}
	\]
	and hence the syndrome class probabilities $Q_{\Syn}(y)=\sum_{\Syn e=y}Q(e)$ (as in Eqs.~(\ref{eq:syndromebetas}) and (\ref{eq:syndromebetas2})). 
	A direct computation shows that the post-measurement state takes the form
	\[
	\cE(\rho)_x = \sum_e (P\ast^{\Syn}Q)_x(e)e\rho e^{\dagger},
	\]
	where
	\[
	(P\ast^{\Syn}Q)(e,x)\coloneqq \sum_f P(ef^\dagger, x+\Syn(f))Q(f)
	\]
	is a `syndrome-twisted' form of convolution over the Pauli group (also see Proposition~\ref{prop:composition_measurements}). Hence 
	\[
		p_x(y)\coloneqq(P\ast^{\Syn}Q)_{\Syn,x}(y)=\frac{\sum_{\Syn e = y}(P\ast^{\Syn}Q)(e,x)}{(P\ast^{\Syn}Q)(x)}
	\]
	describes the conditional probability that an overall error with syndrome $y$ has occurred. By inspection, this only depends on the distribution of the random variable $\Syn$ with respect to $Q$ and $P_u$. Namely, it is the convolution over $\ftwo^m$ (cf.\ Notation~\ref{not:conv})
	\begin{align*}
			(P\ast^{\Syn}Q)_{\Syn}(y,x) &= (P_{Syn}(u, x+y+u)\ast_u Q_{\Syn} (u))(y)	\\
			&= \sum_u\left(\left(\sum_{\Syn e = y + u}P(e, x + u)\right)\left(\sum_{\Syn e = u}Q(e)\right)\right).
	\end{align*}
\end{rem}

\section{The hardware graph of the ibm\_hanoi backend}
\label{app:hardwaregraph}
The connectivity of the used IBM backends is given by a `Heavy Hexagon'-graph, see Fig.~\ref{fig:backend}. The nodes correspond to qubits and the edges indicate the possibility to apply CNOT gates to the connected qubits. The maximum valence of any qubit is 3.
\begin{figure}[tp]
	\centering
	\begin{overpic}[width=0.5\linewidth]{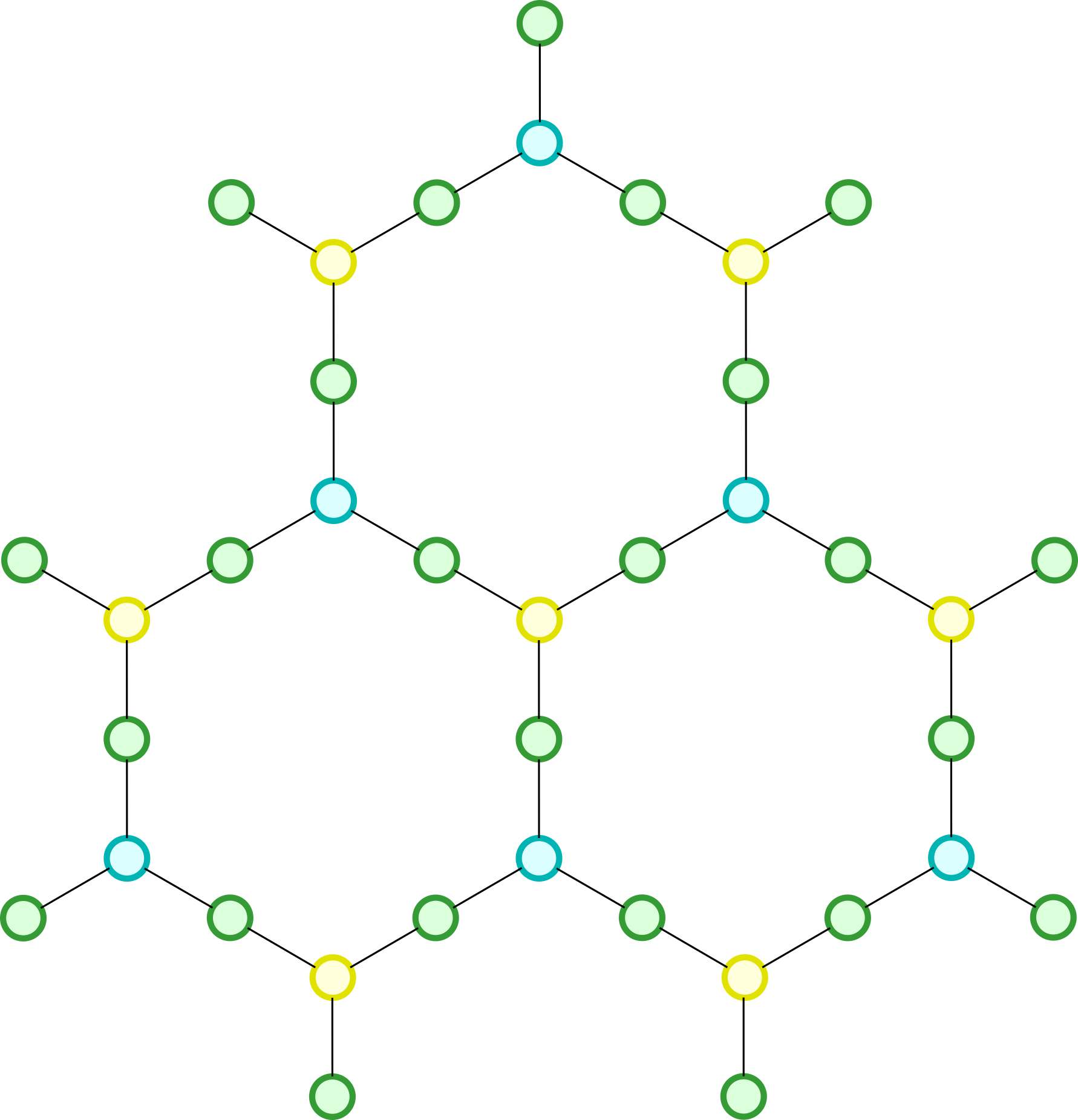}
		
	\end{overpic}
	\caption{The `Heavy Hexagon'-graph giving the topology of the IBM backends. The nodes correspond to qubits and the edges indicate the possibility to apply CNOT gates to the connected qubits. The color only highlights the structure of the graph.}
	\label{fig:backend}
\end{figure}

We need to find a way to map circuit-qubits to the backend and implement operations that realize the stabilizer measurements, while leaving the embedding map invariant, i.e. undoing any necessary swap-operations. One possible way to achieve this is given by the operations schedule in Fig.~\ref{fig:schedule}. 

\begin{figure}[tp]
	\centering
	\begin{overpic}[width=0.5\linewidth]{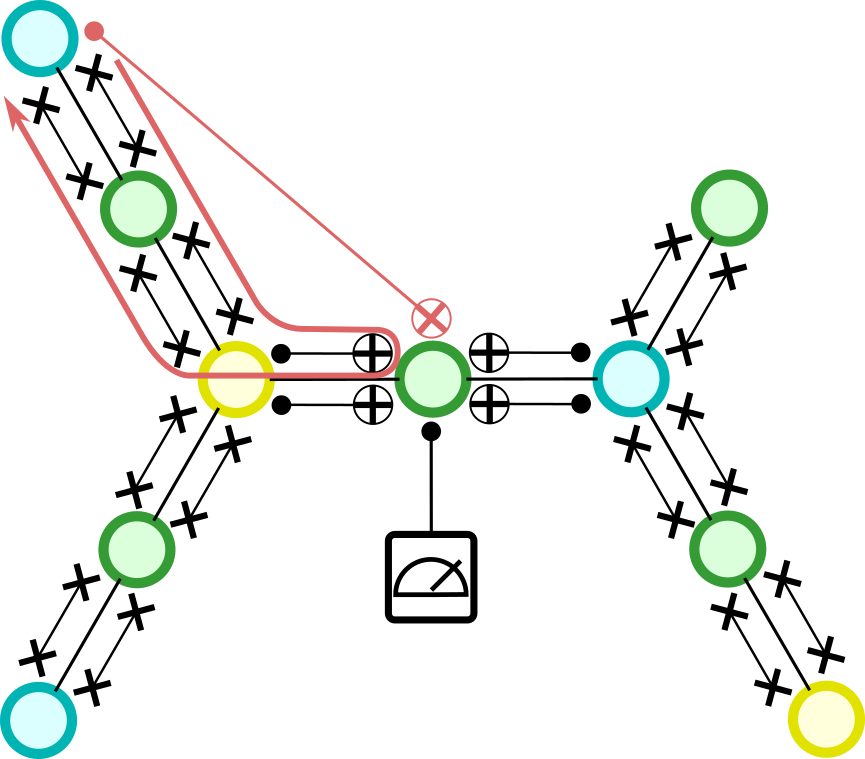}
		\put(2,82){$D_6$}
		\put(13.5,62.5){$D_2$}
		\put(13.5,23){$D_1$}
		\put(2,3.5){$D_5$}
		\put(81.5,62.5){$D_3$}
		\put(81.5,23){$D_0$}
		\put(93,3.5){$D_4$}
		\put(25,43){$A_1$}
		\put(47.5,43){$A_0$}
		\put(70,43){$A_2$}
	\end{overpic}
	\caption{An operations schedule to implement the measurement of the stabilizer operator $S_4 = \1\1\1 ZZZZ$. The operations are carried out in the order indicated by the red arrow, which in sum realize the CNOT gate likewise indicated in red. The information stored in the data qubits $D_i$ is swapped to one of the `edge' ancillas $A_1$ or $A_2$, which is then coupled via a CNOT operation with the central ancilla $A_0$. Afterwards, the swap-operations are undone to recover the original qubit embedding. Finally, measuring the ancilla qubit $A_0$ implements the stabilizer measurement.}
	\label{fig:schedule}
\end{figure}

To implement this procedure on the real backend, we map the topology of a quantum circuit to the backend, then, in case there are multiple possible matches, choose one such that the total CNOT error along the embedded graph is minimized.

\section{Calibration and correction algorithms}
\label{app:algos}

\begin{algorithm}[H]
	\caption{Calibration}\label{alg:calibration}
	\begin{algorithmic}[1]
		\linespread{1.5}\selectfont
		\Procedure{CalibrationExperiment}{$m,\ket{\psi}, \mathcal{C}$} \Comment{See also Fig.~\ref{fig:calibration}.}
			\State \textbf{Input:} number of bits in the measurement outcome $m$, input state $\ket{\psi}$, Syndrome measurement circuit $\mathcal{C}$.
			\State \textbf{Output: } Measurement statistics for the first and second measurement in the calibration experiment.
			\State $\rho_{\mathrm{cal.}} \gets \proj{\psi}$
			\State $P\gets \Call{OutcomeDistribution}{\textbf{run}\; \mathcal{C}^2(\rho_{\mathrm{cal}})}$ \Comment{probability distribution of the results of the two syndrome measurements}
			\ForAll{$x\in\{0,1,...,2^m-1\}$}
				\State $\tilde{p}(x|\rho_{\mathrm{cal.}}) \gets$ $P$(1st outcome = $x$)
				\ForAll{$y\in\{0,1,...,2^m-1\}$}
					\State $\tilde{p}(y|\rho_{\mathrm{cal.},x}) \gets$ $P$(2nd outcome = $y$ $\vert$ 1st outcome = $x$)		
				\EndFor
			\EndFor
			\State\Return $(\tilde{p}(x|\rho_{\mathrm{cal}}),\; \tilde{p}(y|\rho_{\mathrm{cal.},x}))$
		\EndProcedure\vspace*{3ex}
		\Procedure{CalibrationParameters}{$m,\tilde{p}(x|\rho_{\mathrm{cal}}), \tilde{p}(y|\rho_{\mathrm{cal},x}), S(a)$}
		\State \textbf{Input:}  number of bits in the measurement outcome $m$, the measurement statistics of the calibration experiment $\tilde{p}(x|\rho_{\mathrm{cal}})$ and $\tilde{p}(y|\rho_{\mathrm{cal},x})$,
		the stabilizer elements $S(a)$ for all $a\in\{0,1,...,2^m-1\}$.
		\State \textbf{Output:} The calibration parameters $\alpha_{S(a)}$, $\beta_{S(a)}$, $\beta_{S(a),u}$, $\gamma_{S(a)}$.
		\State $(\tilde{p}(x|\rho_{\mathrm{cal}}), \tilde{p}(y|\rho_{\mathrm{cal},x}))  \gets \Call{CalibrationExperiment}{m,\ket{\psi}, \mathcal{C}}$
		
		\ForAll{$a\in\{0,1,...,2^m-1\}$}
		\State $\langle S(a)\rangle_{\mathrm{cal.}} \gets \bra{\psi} S(a) \ket{\psi}$ \Comment{Calculate ideal value analytically}
		\State $\widetilde{\langle S(a)\rangle}_{\mathrm{cal.}}^{(1)} \gets \sum_{x=0}^{2^m-1} (-1)^{x\cdot a} \tilde{p}(x|\rho_{\mathrm{cal.}})$ \Comment{1st expectation value from measured data}
		\ForAll{$x\in\{0,1,...,2^m-1\}$}
		\State $\widetilde{\langle S(a)\rangle}_{\rho_{\mathrm{cal.},x}} \gets \sum_{y=0}^{2^m-1} (-1)^{y\cdot a} \tilde{p}(y|\rho_{\mathrm{cal.,x}})$ \Comment{2nd expectation value depending on outcome} 
		\EndFor
		\State $\widetilde{\langle S(a)\rangle}_{\mathrm{cal.}}^{(2)} \gets \sum_{x=0}^{2^m-1} \tilde{p}(x|\rho_{\mathrm{cal.}}) \widetilde{\langle S(a)\rangle}_{\rho_{\mathrm{cal.},x}}$ \Comment{2nd expectation value, averaged over outcomes}
		\EndFor
		\ForAll{$a\in\{0,1,...,2^m-1\}$}
		\State $\beta_{S(a)} \gets \frac{\widetilde{\langle S(a)\rangle}_{\mathrm{cal.}}^{(2)}}{\widetilde{\langle S(a)\rangle}_{\mathrm{cal.}}^{(1)}}$ \Comment{See Eq.~(\ref{eq:betagamma})}
		\State $\gamma_{S(a)} \gets \frac{\widetilde{\langle S(a)\rangle}_{\mathrm{cal.}}^{(1)}}{\langle S(a)\rangle_{\mathrm{cal.}}}$
		\State $\alpha_{S(a)} \gets \frac{\beta_{S(a)}}{\gamma_{S(a)}}$
		\Comment{See Eq.~(\ref{eq:alpha})}
		\ForAll{$u\in\{0,1,...,2^m-1\}$}
		\State $\beta_{S(a),u} \gets \frac{1}{2^m}\sum_{b,x}(-1)^{(u\oplus x)\cdot b}\frac{\widetilde{\langle S(a)\rangle}_{\rho_{\mathrm{cal.},x}} \tilde{p}(x\vert\rho_{\mathrm{cal.}})}{\gamma_{S(a)} \langle S(a\oplus b) \rangle_{\mathrm{cal.}} }$
		\Comment{See Eq.~(\ref{eq:betaSau_cal})}
		\EndFor
		\EndFor			
		\State \Return $(\alpha_{S(a)}, \beta_{S(a)}, \beta_{S(a),u}, \gamma_{S(a)} )$
		\EndProcedure
	\end{algorithmic}
\end{algorithm}	

\begin{figure}
	\begin{algorithm}[H]
		\caption{Syndrome Measurement Correction (Averaged)}\label{alg:correctionaverage}
		\begin{algorithmic}[1]
			\linespread{1.5}\selectfont
			\Procedure{CorrectedExpectationValue}{$\tilde{p}(x|\rho)$, $S(a)$, $m$, $\alpha_{S(a)}$}
			\State \textbf{Input:} the observed probabilities for every outcome $\tilde{p}(x|\rho)$, the stabilizer element $S(a)$, the number of bits in the outcome $m$, the calibration parameters $\alpha_{S(a)}$.
			\State \textbf{Output:} The ideal expectation value of $S(a)$ after the measurement. 
			\State $\widetilde{\langle S(a)\rangle}_{\rho} \gets \sum_{x=0}^{2^m-1} (-1)^{x\cdot a} \tilde{p}(x|\rho)$
			\State	$\langle S(a)\rangle_{\rho} \gets \alpha_{S(a)} \widetilde{\langle S(a)\rangle}_{\rho} $
			\Comment{See Eq.~(\ref{eq:correctedexpectationvalue})}
			\State \Return $\langle S(a)\rangle_{\rho}$
			\EndProcedure
		\end{algorithmic}
	\end{algorithm}	
\end{figure}

\begin{figure}
	\begin{algorithm}[H]
		\caption{Syndrome Measurement Correction}\label{alg:correction}
		\begin{algorithmic}[1]
			\linespread{1.5}\selectfont
			\Procedure{CorrectedExpectationValue}{$x$, $\tilde{p}(x|\rho)$, $S(a)$, $m$, $\beta_{S(a),u}$}
				\State \textbf{Input:} the observed syndrome $x$, the observed probabilities for every outcome $\tilde{p}(x|\rho)$, the stabilizer element $S(a)$, the number of bits in the outcome $m$, the calibration parameters $\beta_{S(a),u}$.
				\State \textbf{Output:} The ideal expectation value of $S(a)$ after the measurement. 
				\State $\widetilde{\langle S(a)\rangle}_{\rho} \gets \sum_{x=0}^{2^m-1} (-1)^{x\cdot a} \tilde{p}(x|\rho)$
				\State $\langle S(a)\rangle_{\rho} \gets \frac{1}{\gamma_{S(a)}}\widetilde{\langle S(a)\rangle}_{\rho}$
				\State $p(x|\rho) = \frac{1}{2^m} \sum_{b=0}^{2^m-1} \langle S(b) \rangle_{\rho} (-1)^{x\cdot b}$
				\State	$\langle S(a)\rangle_{\rho_x}
				\gets \frac{1}{\tilde{p}(x|\rho)}\sum_{u=0}^{2^m-1} (-1)^{a\cdot(x\oplus u)} p(x\oplus u|\rho)\beta_{S(a),u}$
				\Comment{See Eq.~(\ref{eq:expectationvaluegivenoutcome})} 
				\State \Return $\langle S(a)\rangle_{\rho_x}$
			\EndProcedure
		\end{algorithmic}
	\end{algorithm}	
\end{figure}

\section{Noise estimation for the $[[7,1,3]]$ Steane code}
\label{app:Steane}

\begin{figure*}
	\centering
	\begin{quantikz}[row sep=1pt, column sep=1pt]
		\lstick{1}      & \qw     & \qw      & \qw      & \qw       & \qw       & \targa{}   & \qw       & \qw       & \qw       & \qw       & \qw       & \qw       & \qw       & \qw       & \qw       &\qw     &\qw&\qw & \qw&\qw      & \qw     & \qw      & \qw      & \qw       & \qw       & \phasea{}  & \qw       & \qw       & \qw       & \qw       & \qw       & \qw       & \qw       & \qw       & \qw       &\qw     &\qw\\
		\lstick{2}      & \qw     & \qw      & \qw      & \qw       & \qw       & \qw       & \qw       & \phasec{}  & \qw       & \qw       & \qw       & \qw       & \qw       & \qw       & \qw       &\qw     &\qw&\qw&\qw&\qw      & \qw     & \qw      & \qw      & \qw       & \qw       & \qw       & \qw       & \targc{}   & \qw       & \qw       & \qw       & \qw       & \qw       & \qw       & \qw       &\qw     &\qw\\
		\lstick{3}      & \qw     & \qw      & \qw      & \qw       & \qw       & \qw       & \qw       & \qw       & \targa{}   & \qw       & \qw       & \qw       & \qw       & \qw       & \phasec{}  &\qw     &\qw&\qw&\qw&\qw      & \qw     & \qw      & \qw      & \qw       & \qw       & \qw       & \qw       & \qw       & \phasea{}  & \qw       & \qw       & \qw       & \qw       & \qw       & \targc{}   &\qw     &\qw\\
		\lstick{4}      & \qw     & \qw      & \qw      & \qw       & \qw       & \qw       & \qw       & \qw       & \qw       & \phaseb{}  & \qw       & \qw       & \qw       & \qw       & \qw       &\qw     &\qw&\qw& \qw&\qw      & \qw     & \qw      & \qw      & \qw       & \qw       & \qw       & \qw       & \qw       & \qw       & \targb{}   & \qw       & \qw       & \qw       & \qw       & \qw       &\qw     &\qw& & &\\
		\lstick{5}      & \qw     & \targa{}  & \qw      & \qw       & \qw       & \qw       & \phaseb{}  & \qw       & \qw       & \qw       & \qw       & \qw       & \qw       & \qw       & \qw       &\qw     &\qw&\qw&\qw&\qw      & \qw     & \phasea{} & \qw      & \qw       & \qw       & \qw       & \targb{}   & \qw       & \qw       & \qw       & \qw       & \qw       & \qw       & \qw       & \qw       &\qw     &\qw\\
		\lstick{6}      & \qw     & \qw      & \qw      & \phasec{}  & \qw       & \qw       & \qw       & \qw       & \qw       & \qw       & \qw       & \qw       & \qw       & \phaseb{}  & \qw       &\qw     &\qw&\qw&\qw&\qw     & \qw     & \qw      & \qw      & \targc{}   & \qw       & \qw       & \qw       & \qw       & \qw       & \qw       & \qw       & \qw       & \qw       & \targb{}   & \qw       &\qw     &\qw& & &\\
		\lstick{7}      & \qw     & \qw      & \phaseb{} & \qw       & \qw       & \qw       & \qw       & \qw       & \qw       & \qw       & \phasec{}  & \qw       & \targa{}   & \qw       & \qw       &\qw     &\qw&\qw&	\qw&\qw      & \qw     & \qw      & \targb{}  & \qw       & \qw       & \qw       & \qw       & \qw       & \qw       & \qw       & \targc{}   & \qw       & \phasea{}  & \qw       & \qw       &\qw     &\qw\\
		\lstick{\ket{0}}& \gate[style={draw=color1}, label style={color1}]{H}\qwa& \ctrla{-3}\qwa& \qwa      & \qwa      & \phase{}\qwa  & \ctrla{-7}\qwa & \qwa       & \qwa       & \ctrla{-5}\qwa & \qwa       & \qwa       & \phase{}\qwa  & \ctrla{-1}\qwa & \qwa       & \qwa       &\gate[style={draw=color1}, label style={color1}]{H}\qwa&\meter[draw=color1]{}\qwa\rstick{$S_3$} &\hphantom{S_3\;\ket{0}}  & & \lstick{\ket{0}}& \gate[style={draw=color1}, label style={color1}]{H}\qwa& \ctrla{-3}\qwa& \qwa      & \qwa      & \phase{}\qwa  & \ctrla{-7}\qwa & \qwa       & \qwa       & \ctrla{-5}\qwa & \qwa       & \qwa       & \phase{}\qwa  & \ctrla{-1}\qwa & \qwa       & \qwa       &\gate[style={draw=color1}, label style={color1}]{H}\qwa&\meter[draw=color1]{}\qwa\rstick{$S_6$}\\
		\lstick{\ket{0}}& \gate[style={draw=color2}, label style={color2}]{H}\qwb& \qwb      & \ctrlb{-2}\qwb& \qwb       & \qwb       & \qwb       & \ctrlb{-4}\qwb & \qwb       & \qwb      & \ctrlb{-5}\qwb & \qwb       & \ctrl{-1}\qwb & \qwb       & \ctrlb{-3}\qwb & \qwb       &\gate[style={draw=color2}, label style={color2}]{H}\qwb&\meter[draw=color2]{}\qwb\rstick{$S_4$}& & & \lstick{\ket{0}}& \gate[style={draw=color2}, label style={color2}]{H}\qwb& \qwb      & \ctrlb{-2}\qwb& \qwb       & \qwb       & \qwb       & \ctrlb{-4}\qwb & \qwb       & \qwb      & \ctrlb{-5}\qwb & \qwb       & \ctrl{-1}\qwb & \qwb       & \ctrlb{-3}\qwb & \qwb       &\gate[style={draw=color2}, label style={color2}]{H}\qwb&\meter[draw=color2]{}\qwb\rstick{$S_1$}\\
		\lstick{\ket{0}}& \gate[style={draw=color3}, label style={color3}]{H}& \qwc      & \qwc      & \ctrlc{-4}\qwc & \ctrl{-2}\qwc & \qwc       & \qwc       & \ctrlc{-8}\qwc & \qwc       & \qwc       & \ctrlc{-3}\qwc & \qwc       & \qwc       & \qwc       & \ctrlc{-7}\qwc &\gate[style={draw=color3}, label style={color3}]{H}\qwc&\meter[draw=color3]{}\qwc\rstick{$S_5$}& & & \lstick{\ket{0}}& \gate[style={draw=color3}, label style={color3}]{H}& \qwc      & \qwc      & \ctrlc{-4}\qwc & \ctrl{-2}\qwc & \qwc       & \qwc       & \ctrlc{-8}\qwc & \qwc       & \qwc       & \ctrlc{-3}\qwc & \qwc       & \qwc       & \qwc       & \ctrlc{-7}\qwc &\gate[style={draw=color3}, label style={color3}]{H}\qwc&\meter[draw=color3]{}\qwc\rstick{$S_2$}
	\end{quantikz}
	\caption{A circuit to extract the syndrome of the 7-qubit-Steane code, cf. \cite{Reichardt2020}. The black gates compensate for the different ordering compared to the standard sequential syndrome measurement.}\label{fig:flagqubitcircuits}	
\end{figure*}

Here we detail the estimation of the noise channel from the stabilizer measurements according to the method of Wagner \textit{et al.}~\cite{wagner2022pauli} for the  $[[7,1,3]]$ Steane code, with the syndrome measurement as implemented in Fig.~\ref{fig:flagqubitcircuits}. Our choice of the six generators of the stabilizer is shown in Table~\ref{tab:steane}.
\begin{table}[tp]%
	\caption{The stabilizer generators and logical operators of the 7-qubit-Steane code~\cite{Steane96b}.}\label{tab:steane}%
	\centering%
	\begin{tabular}{c|ccccccc}
		$S_1$ & $\1$ & $\1$ & $\1$ & $X$  & $X$  & $X$  & $X$\\
		$S_2$ & $\1$ & $X$  & $X$  & $\1$ & $\1$ & $X$  & $X$\\
		$S_3$ & $X$  & $\1$ & $X$  & $\1$ & $X$  & $\1$ & $X$\\
		$S_4$ & $\1$ & $\1$ & $\1$ & $Z$  & $Z$  & $Z$  & $Z$\\
		$S_5$ & $\1$ & $Z$  & $Z$  & $\1$ & $\1$ & $Z$  & $Z$\\
		$S_6$ & $Z$  & $\1$ & $Z$  & $\1$ & $Z$  & $\1$ & $Z$\\[1ex]
		$\bar{X}$ & $X$ & $X$ & $X$ & $X$ & $X$ & $X$ & $X$\\
		$\bar{Z}$ & $Z$ & $Z$ & $Z$ & $Z$ & $Z$ & $Z$ & $Z$
	\end{tabular}%
\end{table}%

We recall the form of the probability mass function $P\colon \mathbb P_7\rightarrow [0,1]$ for a Pauli channel with independent Pauli errors on individual qubits:
\begin{equation}
	P(e_1\otimes \cdots\otimes e_7)=P_1(e_1)\cdots P_7(e_7),\quad e_i\in \{\1, X, Y, Z\}
\end{equation}

Given a stabilizer state affected by the error channel, we also have the multiplicative relation for the expectation values
\begin{equation}
	\langle e_1\otimes \cdots\otimes e_7 \rangle = \langle e_1 \rangle\cdots \langle e_7 \rangle,
\end{equation}
where on the right-hand side we interpret $e_i\in \mathbb P_7$ by acting as the identity on the qubits
not indexed by $i$.
We recall the Hamming-matrix
\begin{equation}
	H_{[7,4,3]}=\begin{pmatrix}
		0 & 0 & 0 & 1 & 1 & 1 & 1 \\
		0 & 1 & 1 & 0 & 0 & 1 & 1 \\
		1 & 0 & 1 & 0 & 1 & 0 & 1
	\end{pmatrix}
\end{equation}
associated with the Steane code (for X and Z errors separately) and use linear combinations over $\ftwo$ of the rows to construct the invertible matrix (over $\mathbb Q$)
\begin{equation}
	D=\left(
	\begin{array}{ccccccc}
		0 & 0 & 0 & 1 & 1 & 1 & 1 \\
		0 & 1 & 1 & 0 & 0 & 1 & 1 \\
		0 & 1 & 1 & 1 & 1 & 0 & 0 \\
		1 & 0 & 1 & 0 & 1 & 0 & 1 \\
		1 & 0 & 1 & 1 & 0 & 1 & 0 \\
		1 & 1 & 0 & 0 & 1 & 1 & 0 \\
		1 & 1 & 0 & 1 & 0 & 0 & 1 \\
	\end{array}
	\right),
\end{equation}
where the binary representation of the row-index determines the selection from $H$.
The expectation values of $X$-operators are then related to the ones of the stabilizer elements via $D$. Under the assumption that the latter values are positive, this takes the form of a uniquely solvable linear equation among the logarithms:
\begin{equation}
	\begin{aligned}
	&\ln  (\langle X_1 \rangle, \langle X_2 \rangle, \langle X_3 \rangle, \langle X_4 \rangle, \langle X_5 \rangle, \langle X_6 \rangle, \langle X_7 \rangle )^T\\
	 =& D^{-1} \ln(\langle S_1 \rangle, \langle S_2 \rangle,\langle S_1 S_2 \rangle, \langle S_3 \rangle, \langle S_1 S_3 \rangle, \langle S_2 S_3 \rangle, \langle S_1 S_2 S_3 \rangle)^T.
	\end{aligned}
\end{equation}
Similar equations hold for $Z$ and $Y$ with $(S_1, S_2, S_3)$ being replaced by $(S_4, S_5, S_6)$ respectively $(S_1S_4, S_2S_5, S_3S_6)$.
Then for the $i$-th qubit we obtain the single qubit error rates
\begin{equation}
	\begin{aligned}
		P_i(I) = \frac{1}{4}(1 + \langle X_i \rangle +\langle Y_i \rangle+\langle Z_i \rangle )\\
		P_i(X) = \frac{1}{4}(1 + \langle X_i \rangle -\langle Y_i \rangle-\langle Z_i \rangle )\\
		P_i(Y) = \frac{1}{4}(1 - \langle X_i \rangle +\langle Y_i \rangle-\langle Z_i \rangle )\\
		P_i(Z) = \frac{1}{4}(1 - \langle X_i \rangle -\langle Y_i \rangle+\langle Z_i \rangle )
	\end{aligned} 
\end{equation}
and have thus reconstructed $P$.

When the assumption of independent noise on each qubit is not fulfilled, one can of course still form the `probability' vector
\begin{equation}
	\vec{p}(\langle S(0)\rangle ,\langle S(1)\rangle ,\ldots, \langle S(63)\rangle ) = \otimes_{i=1}^7 (1-P_i(X)-P_i(Y)-P_i(Y),P_i(X),P_i(Y),P_i(Z)) \in \mathbb R^{4^7}
	\label{eq:reconstructedP}
\end{equation}
given by the Kronecker product,
as long as the values on the right-hand side above are defined and with the caveat that the $4$-tuples might not lie in the probability simplex. The effect is small in this example and we chose to fix it by clamping the values to the interval [0, 1] and renormalizing the distribution. Here we added the arguments to the vector $\vec{p}$ to emphasize that the estimation of the Pauli channel is a function of the expectation values of the stabilizer elements (actually only the 21 specific ones appearing above).

\end{document}